%% file: 000main.tex
\documentclass[conference]{IEEEtran}
\IEEEoverridecommandlockouts
\makeatletter
\let\NAT@parse\undefined
\makeatother
\usepackage{hyperref}  

\usepackage{cite}
\usepackage{amsmath,amssymb,amsfonts}
\usepackage{graphicx}
\usepackage{textcomp}
\usepackage{xcolor}
\usepackage{threeparttable}
\usepackage{caption}

\usepackage{booktabs} 
\usepackage{array}
\usepackage{graphicx}
\usepackage{algpseudocode}
\usepackage{amsthm}
\usepackage{multirow}
\usepackage{tabularx}
\usepackage{booktabs}
\usepackage{bm}
\usepackage{algorithm}
\usepackage{subcaption}
\usepackage{enumitem}
\newcommand{\squishlist}{
 \begin{list}{∙\bullet}
  { \setlength{\itemsep}{1pt}
     \setlength{\parsep}{1pt}
     \setlength{\topsep}{2pt}
     \setlength{\partopsep}{1pt}
     \setlength{\leftmargin}{1.5em}
     \setlength{\labelwidth}{1em}
     \setlength{\labelsep}{0.5em} } }

\newcommand{\squishlisttwo}{
 \begin{list}{∙\bullet}
  { \setlength{\itemsep}{0pt}
     \setlength{\parsep}{1pt}
    \setlength{\topsep}{1pt}
    \setlength{\partopsep}{0pt}
\setlength{\leftmargin}{pt} } }

\newcommand{\squishend}{
\end{list}  }

\newcolumntype{N}{@{}m{0pt}@{}}

\makeatletter
\newcommand{\multiline}[1]{%
  \begin{tabularx}{\dimexpr\linewidth-\ALG@thistlm}[t]{@{}X@{}}
    #1
  \end{tabularx}
}
\makeatother

\newcommand{\hide}[1]{}

\let\oldReturn\Return
\renewcommand{\Return}{\State\oldReturn}

\setlist[itemize]{leftmargin=10pt}

\theoremstyle{plain}
\newtheorem{lemma}{Lemma}

\newtheorem{problem}{Problem}
\graphicspath{{figures/}}
\DeclareMathOperator*{\argmax}{argmax}

\theoremstyle{definition}

\newcommand{\genius}{\textsc{Genius}}
\def\BibTeX{{\rm B\kern-.05em{\sc i\kern-.025em b}\kern-.08em
    T\kern-.1667em\lower.7ex\hbox{E}\kern-.125emX}}
    
\begin{document}
\title{Conference Paper Title*\\
{\footnotesize \textsuperscript{*}Note: Sub-titles are not captured in Xplore and
should not be used}
\thanks{Identify applicable funding agency here. If none, delete this.}
}

\title{\genius: A Novel Solution for Subteam Replacement with Clustering-based Graph Neural Network\\
}

\author{\IEEEauthorblockN{Chuxuan Hu,
Qinghai Zhou, Hanghang Tong}
\IEEEauthorblockA{University of Illinois at Urbana-Champaign \\
\{chuxuan3, qinghai2, htong\}@illinois.edu}}

\maketitle
\input{00abstract.tex}

\input{01introduction.tex}
\input{03problem_def.tex}

\input{04method.tex}
\input{05experiment.tex}
\input{02related_work.tex}
\input{06conclusion}
\input{07ack}
\bibliographystyle{IEEEtrans}
\bibliography{reference}

\end{document}

%% file: 00abstract.tex
\begin{abstract}

\textit{Subteam replacement} is defined as finding the optimal candidate set of people who can best function as an unavailable subset of members (i.e., subteam) for certain reasons (e.g., conflicts of interests, employee churn), given a team of people embedded in a social network working on the same task. 
Prior investigations on this problem incorporate graph kernel as the optimal criteria for measuring the similarity between the new optimized team and the original team. However, the increasingly abundant social networks reveal fundamental limitations of existing methods, including (1) the graph kernel-based approaches are powerless to capture the key intrinsic correlations among node features, (2) they generally search over the entire network for every member to be replaced, making it extremely inefficient as the network grows, and (3) the requirement of equal-sized replacement for the unavailable subteam can be inapplicable due to limited hiring budget. In this work, we address the limitations in the state-of-the-art for \emph{subteam replacement} by (1) proposing \textsc{Genius}, a novel clustering-based graph neural network (GNN) framework that can capture team network knowledge for flexible \emph{subteam replacement}, and (2) equipping the proposed \textsc{Genius} with self-supervised \textit{positive team contrasting} training scheme to improve the team-level representation learning and unsupervised node clusters to prune candidates for fast computation. Through extensive empirical evaluations, we demonstrate the efficacy of the proposed method (1) 
\emph{effectiveness}: being able to select better candidate members that significantly increase the similarity between the optimized and original teams, and (2) \emph{efficiency}: achieving more than $\times 600$ speed-up in average running time.


\end{abstract}

%% file: 01introduction.tex
\section{Introduction}\label{sec:intro}
The emergence of network science of teams has revolutionized the way to characterize, predict and optimize teams that are embedded in large-scale social networks. Among others, given a team of people working on the same task (e.g., launching a targeted web service, proceeding towards a specific research topic), \textit{subteam replacement} is dedicated to finding the optimal candidate set of people who can best perform the function of a subset of the original team (i.e., subteam), if the subteam becomes unavailable for certain reasons. Real-world applicable scenarios include (1) in a software development team, a subteam of employees might be assigned to other departments due to new business requirements, (2) in artistic groups like choirs or dance troupes, a subgroup of artists might be unable to show up in a play due to schedule and/or staging conflicts, and many more. In these situations, the goal of \emph{subteam replacement} is to address the potential deterioration of the team performance due to the absence of subteam members.

State-of-the-art adopts the following two principles in designing \emph{subteam replacement} algorithms. First (\textit{structural similarity}), to reduce the disruption to the current team, the recommended candidates should have a similar collaboration relationship as the current team members in terms of the connections with the remaining team members. Second (\textit{skill similarity}), the new members should possess a similar skill set as the members from the unavailable subteam, therefore they can best fulfill the requirements of the team-level task. Most, if not all, of the existing works leverage random walk graph kernel~\cite{vishwanathan2010graph} as the cornerstone for \textit{subteam replacement} to capture both the structural and skill match, where the key idea is to select the candidate members that can achieve the largest similarity between the newly constructed and the original teams~\cite{Liangyue2015, Zhaoheng2021}. 

Although these approaches are effective in recommending well-qualified candidates in \emph{subteam replacement}, there are three major limitations with the graph kernel-based approaches. First, from the perspective of effectiveness, in quantifying the similarity between teams, graph kernel-based methods separately exploit the attributes (i.e., skills), which inevitably ignores the potential correlations between various skills. For example, in research team replacement (e.g., on DBLP dataset), the researcher's skills correspond to the research areas in computer science (e.g., data mining, NLP, machine learning, etc.). Graph kernel-based methods will separately evaluate the closeness between two teams regarding an individual research area (i.e., skill), without considering the essential correlation among different areas (e.g., machine learning and data mining), which will sabotage the performance of graph kernel-based approaches. It is worth noting that we observe a drastic performance drop w.r.t. multiple graph similarity as the skill set scales up (e.g., keywords in publications) and we have a detailed discussion in Sec.~\ref{sec:method} and~\ref{sec:exp}. Second, graph kernel-based methods are computationally intractable in terms of time and space complexity since they iteratively calculate the similarity between the original team and the new team with candidate members from an enormous search space. Furthermore, existing \emph{subteam replacement} methods require the same size of candidates as the unavailable subteam. Nonetheless, in real-world scenarios, recommending a compact set of members that can undertake the same tasks would be more preferable due to limited hiring budget.



In this paper, we carefully address the aforementioned limitations and propose a novel graph neural network (GNN) framework based on clustering, namely \genius, to efficiently recommend well-qualified candidates for \emph{subteam replacement}. Specifically, given the team network embedded in a large social network, the proposed \textsc{Genius} framework first utilizes a GNN-backboned team representation learner to encode the knowledge of collaboration structure and skill in member-level embeddings, which will be employed for accurate \emph{subteam replacement}. To further improve the representation learning, we design a self-supervised \textit{positive team contrasting} 
scheme where the key idea is to enforce the similarity between the subteam representation and that of the corresponding team. From the perspective of optimization, we aim to minimize the disparity between the recommended subteam and the original team in both collaboration structure and skills. In addition, we further propose a novel clustering-based method, which (1) can speed up the \emph{subteam replacement} by pruning the unqualified candidates, 
and (2) allows flexible sizes of recommended subteam without harming its performance to meet the requirements of cases where smaller new teams are expected. 
To summarize, our main contributions are three-fold:

\begin{itemize}
    \item \textbf{Problem:} we formally define the \emph{subteam replacement} problem. Given a set of unavailable members (i.e., subteam), the key idea is to recommend flexible-sized team members, that can perform the same functionality as the original members.
    \item \textbf{Algorithms and Analysis:} 
    to our best knowledge, we are the first to introduce a trainable framework \genius, which integrates (1) a GNN-backboned \emph{team network encoder} equipped with a clustering layer, and (2) a \textit{subteam recommender} where a within-cluster search algorithm is performed. 
    In addition, we conduct a comprehensive analysis on the advantages of \textsc{Genius} in capturing the correlations among feature attributes over existing graph kernel-based approaches.
    
    \item \textbf{Empirical Evaluations:} We conduct extensive experiments, including quantitative evaluations w.r.t. multiple graph similarity metrics and case studies on real-world social networks, to demonstrate the superiority 
    of our model over the state of arts.
    
\end{itemize}

%% file: 03problem_def.tex
\section{Problem Definition}

In this section, we formally define the flexible-sized \textit{subteam replacement} problem, after introducing the notations used in this paper.

Table \ref{tab:symbols} summarizes main symbols and notations used in this paper. We use uppercase calligraphic letters for sets (e.g.,
$\mathcal{G}$), bold uppercase letters for matrices (e.g., $\mathbf{A}$), bold lowercase letters for vectors (e.g., $\mathbf{r}$) and lower case letter for numbers (e.g., $c$). For vector and matrix indexing, we use $\mathbf{A}[i,:]$ and $\mathbf{A}[:, j]$ to represent the $i$-th row and the $j$-th column of $\mathbf{A}$, respectively. Alternatively, we use $\mathbf{A}[\mathcal{T}, :]$ and $\mathbf{A}[:, \mathcal{T}]$ to denote the rows and columns indexed by the elements in set $\mathcal{T}$, respectively. 
\begin{table}[!ht]
	\centering
	  \caption{Symbols and Notations}
	{\begin{tabular}{lc}
    \toprule
    \textbf{Symbols} & \textbf{Definition}  \\
    \midrule
    $\mathbf{A} \in \mathbb{R}^{n \times n}$ & Adjacency matrix of the entire team social network\\
    $\mathbf{X} \in \mathbb{R}^{n \times d}$ & Node feature matrix of the entire team social network\\
    \midrule
    $\mathcal{G}$ & Entire Social Network\\
    $\mathcal{T} \subset \mathcal{G} $  & The original team\\
    $\mathcal{R} \subset \mathcal{T}$ & The subteam to be replaced\\
    $\mathcal{M} = \mathcal{G} \setminus \mathcal{T}$ & Remaining graph\\
    $\mathcal{S} \subset \mathcal{M}$ & The optimal new subteam\\
    \midrule
    $d, c$ & \#node features, \#clusters, respectively\\
    $n$ & Size of the entire team social network\\
    \midrule
    $\hat{\mathbf{A}}$ & row-wised normalized matrix\\
    $\mathbf{A}^T$ & transposed matrix\\
    \bottomrule
    \label{tab:symbols}
	\end{tabular}
	}
	\label{tab:data_stats}
\end{table}
In general, \emph{subteam replacement} aims to find a set of candidate members to replace the unavailable members in the original team, so that the new team can perform the same functionalities as the original one. In this paper, we denote the entire team social network, in which the teams are embedded, as  $\mathcal{G}=\{\mathbf{A}, \mathbf{X}\}$. $\mathbf{A}$ and $\mathbf{X}$ represent the adjacency and the node attribute matrix of the social network, respectively. The original team of a group of individual members is represented by $\mathcal{T}$ and the subteam that is going to be replaced is denoted as $\mathcal{R}$. Different from the existing works that find the same number of new members (i.e., $|\mathcal{S}|$)~\cite{Zhaoheng2021}, we investigate the flexible-sized replacement by recommending new team members $\mathcal{S}$ where $|\mathcal{S}|\leq|\mathcal{T}|$ from the remaining individuals of $\mathcal{G}$ (i.e.,  $\mathcal{M}=\mathcal{G}\setminus\mathcal{T}$), which is more applicable in real-world scenarios (e.g., limited hiring budget).



Therefore, we formally defined the flexible-sized \textit{subteam seplacement} problem as follows:
\begin{problem}{\textbf {Subteam Replacement}}\label{prob:def}
	\begin{description}
		\item[Given:] (1) a social network $\mathcal{G} = ( \mathbf{A}, \mathbf{X} )$;
		(2) an original team, $\mathcal{T} = \{ \mathbf{A}[\mathcal{T},\mathcal{T}], \mathbf{X}[\mathcal{T},:] \}$; (3) a subteam of individuals to be replaced $\mathcal{R} = \{ \mathbf{A}[\mathcal{R},\mathcal{R}], \mathbf{X}[\mathcal{R},:] \}$

		\item[Goal:] to learn a subteam replacement model, which is capable of capturing the team-level knowledge to replace the unavailable subteam, $\mathcal{R}$, by recommending a new set of candidate members, i.e., $\mathcal{S}$ of flexible size where $|\mathcal{S}|\leq|\mathcal{T}|$. Ideally, the new team, i.e., $\mathcal{T}\setminus\mathcal{R}\cup\mathcal{S}$ should be similar to the original team from the perspectives of collaboration structures and skills.
		
	\end{description}
\end{problem}

%% file: 04method.tex
\section{Proposed Approach}\label{sec:method}

 \begin{figure*}[t!]
    \graphicspath{{figures/}}
    \centering
    \includegraphics[width=0.95\textwidth]{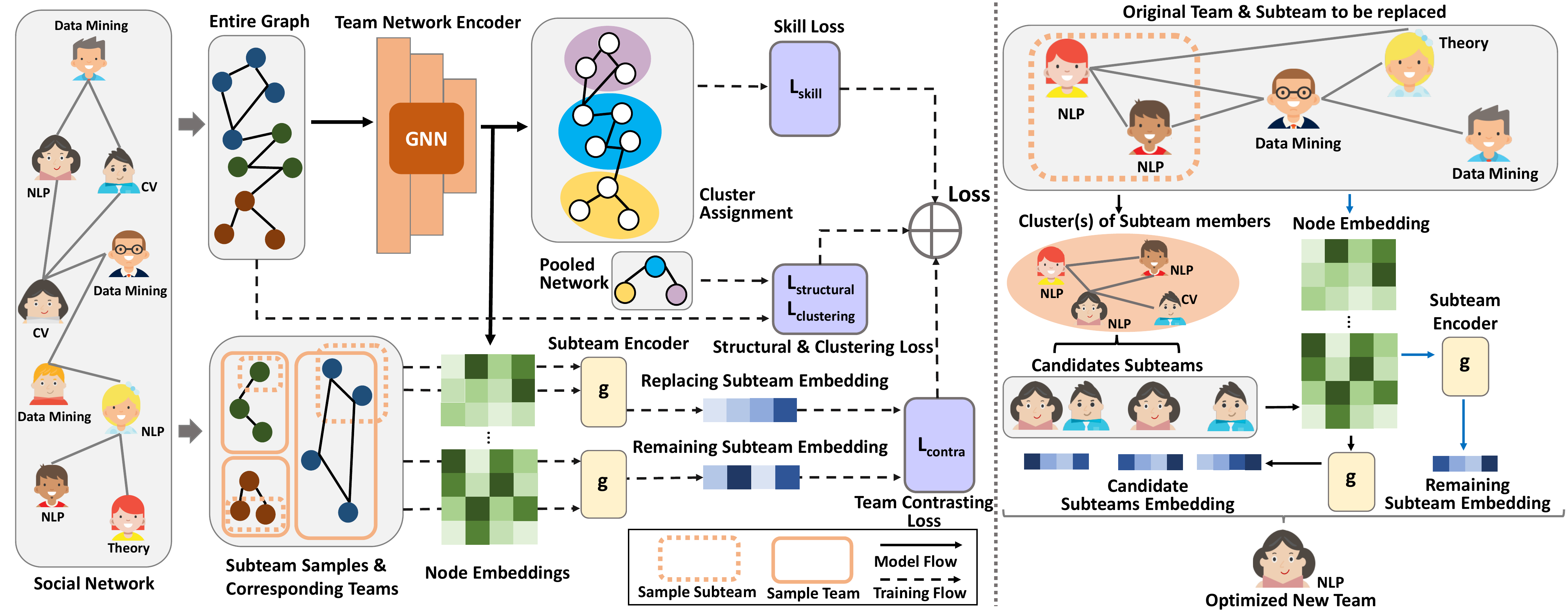}
    \caption{Overview of the \textsc{Genius} Framework. (Left) The data flow of \textit{team network encoder}
    and its training procedure. (Right) Within-cluster search performed by \textit{subteam recommender}.}
    \vspace{-6mm}  
    \label{fig:framework}%
\end{figure*}
In this section, we introduce the details of our proposed framework, \textsc{Genius}, for flexible-sized \emph{subteam replacement}. To be specific, \genius~ addresses the discussed limitations of existing methods by integrating a 
\underline{\textbf{G}}NN-backboned \underline{\textbf{EN}}coder for with\underline{\textbf{I}}n-cl\underline{\textbf{US}}ter subteam search. In Section~\ref{sec:genius_intro}, we first present the two major building blocks of \genius, including (1) \emph{team network encoder} which can capture the essential team social network knowledge to fulfill the requirements of structural and skill match (i.e., training), and (2) \emph{subteam recommender} for recommending candidates members to replace the unavailable subteam members (i.e., inference). We then analyze the superiority of \genius~in obtaining the key correlations among feature attributes over graph kernel-based approaches from the perspectives of \emph{effectiveness} and \emph{efficiency} in section ~\ref{sec:analysis}. An overview of the proposed \textsc{Genius} is provided in Figure~\ref{fig:framework}. 

\vspace{-2mm}
\subsection{\textbf{The \textsc{Genius} Framework}}\label{sec:genius_intro}
\vspace{-1mm}
In \emph{subteam replacement}, we aim to learn the team social network knowledge regarding the collaboration structure (i.e., topology) and skills (i.e., attributes) to recommend a set of well-qualified candidates that can match the characteristics of the original team. To achieve this, we propose a GNN-backboned architecture which is composed of the following components: (1) a \emph{team network encoder}, as the team-level representation learner, and (2) a \emph{subteam recommender}, which is a novel within-cluster searching module for candidates recommendation. We present the details as follows.

\noindent\textbf{Team Network Encoder.} The general idea of \textit{subteam replacement} is to select the candidate members that possess similar characteristics as the original team members in terms of collaborations and skills. Therefore, the replacement performance is considerably dependent on the informativeness of the learned representations of the network. We propose a \emph{team network encoder} which is composed of two sub-components, including (1) a GNN-based team representation learner to capture the correlations between subteams and their corresponding original teams, and (2) a clustering layer
to generate node cluster assignments for efficient member recommendation.




To construct a high-quality \textit{team network encoder}, specifically, we exploit GNNs~\cite{MessagePassing, GCN2016} to map each node in the team social network to a low-dimensional latent space. GNNs are neural network architectures on networked data and can obtain the topological and attribute information of local graph structures through a neighborhood message-passing mechanism. For each intermediate layer, we can compute the node representations as follows:

\begin{equation}
\begin{aligned}
    \mathbf{h}_{i}^l &= \mathrm{F}^l\Big(\{   \mathbf{h}_j^{l-1} \arrowvert \forall j \in \mathcal{N}_i \cup \{v_i\} \}\Big),
    \label{eqn:GNN_update}
\end{aligned}
\end{equation}
\noindent where $\mathbf{h}_i^l$ is the embedding of node $v_i$ at the $l$-th layer and $\mathcal{N}_i$ is the set of one-hop neighbors of node $v_i$. $\mathrm{F}^l(\cdot)$ is a nonlinear function that (1) aggregates the information of neighboring nodes from the previous layer, and (2) updates the representation of a node based on its own embeddings from the previous layer and the combined representation from the neighbors.
 
In \emph{team network encoder}, we stack multiple GNN layers in order to extract the long-range node dependencies in the network, which can be represented as follows:
\begin{equation}
\begin{aligned}
    &\mathbf{H}^{1} = \mathrm{GNN}^1 (\mathbf{A}, \mathbf{X}),\\
    &\dots \\
    &\mathbf{H}^{L} = \mathrm{GNN}^L (\mathbf{A}, \mathbf{H}^{L-1}),
\end{aligned}
\end{equation}
\noindent where $\mathbf{H}^L$ denotes the node embedding matrix at the $L$-th GNN layer. To alleviate the over-smoothing issue in GNNs (i.e., all nodes have similar embeddings) ~\cite{gcn_smoothing}, 
we obtain the final node representations by concatenating the intermediate embeddings from each GNN layer~\cite{JKNet}:
\begin{equation}
      \mathbf{Z} = \mathrm{concat}( \mathbf{H}^{(1)},...,\mathbf{H}^{(L)}),
\end{equation}
where $\mathbf{Z}\in\mathbb{R}^{n\times d}$ denotes learned node embeddings from the \emph{team network encoder}, which encodes the topological and attribute information of the team social network. In \emph{subteam replacement}, the two fundamental principles require that the recommended candidates have (1) a close collaboration relationship (i.e., structural match), and (2) similar attributes (i.e., skill match), as the original team members. Through the proposed GNN-backboned encoder, we can quantitatively evaluate whether the candidate team members are well-qualified in the embedding space by assessing the similarity between the corresponding vector representations. It is worth noting that the \textit{team network encoder} is compatible with any GNN-based architecture~\cite{GCN2016,GraphSAGE,GAT,GPRGNN}, and here we apply Graph Convolutional Networks (GCNs)~\cite{GCN2016} in our implementation.


Furthermore, motivated by the principles in \emph{subteam replacement}, we speculate that the acceptable candidates and the existing team members should constitute a group (i.e., cluster) with similar representations in the embedding space. Hence, we exploit a clustering layer to group the members in the social network according to the obtained representations. Specifically, for a node $v_i$, the clustering layer
first performs a nonlinear transformation on the embedding matrix (i.e., $\mathbf{Z}$) as 
\begin{equation}
\mathbf{E} = \sigma(\mathbf{Z}\mathbf{W}_c),
\end{equation} 
where $\mathbf{W}_c\in\mathbb{R}^{d\times c}$ is the learnable weight matrix and $c$ is the number of clusters. $\sigma$ is the ReLU activation function. The transformed matrix will be used to generate a (soft) cluster assignment matrix $\mathbf{C}\in\mathbb{R}^{n\times c}$ where the elements of the corresponding row $i$ are computed as follows:
\begin{equation}
     \mathbf{C}[i,m] = \frac{e^{\mathbf{E}[i,m]}}{\sum_{j=1}^ce^{\mathbf{E}[i,j]}},
\end{equation}

\noindent where ${m\in\{1,\dots,c\}},\mathbf{C}[i,:]$ is the cluster assignment vector for node $v_i$.



Having obtained the cluster assignment matrix, we further compute the (hard) assignment vector $\mathbf{h} \in \mathbb{R}^{n}$ where
\begin{equation}
     \mathbf{h}[i] = \argmax_{j\in\{1,\dots,c\}}(\mathbf{C}[i,j]),
\end{equation}
We can construct the cluster container $\mathcal{K}$ as:
\begin{align*}
    \mathcal{K}[m] = \{v_i|\mathbf{h}[i] = m, i=\{1,\dots,n\}\}, m\in\{1,\dots,c\},
\end{align*}
\noindent which will be further utilized for recommending candidates.

The clustering layer bears two main advantages. First, it enforces that the nodes with similar embeddings will be grouped together, which fulfills the requirements of structural and skill match in \emph{subteam replacement}. Second, compared to existing methods which adopt an inefficient candidate search over the entire network, the clustering results can be leveraged to significantly speed up the replacement because it largely reduces the search space by pruning the unqualified members.

\smallskip
\noindent\textbf{Training.} The objective of our proposed \genius~ framework is to recommend well-qualified candidate members that are close to the members of existing teams, based on the learned node representations. To navigate the model training, we first propose a self-supervised contrasting scheme, namely \textit{positive team contrasting}, to enhance the node representation learning. In addition, we introduce the training objectives for satisfying both structural and skill match in \emph{subteam replacement}.


\noindent \emph{Positive team contrasting.} Generally speaking, contrastive learning aims to empower the representation learning by maximizing the agreement between similar examples in each pair~\cite{simclr}. In \emph{subteam replacement} problem, we are only provided with the network information from (1) the original teams, and (2) the social network where the teams are embedded. Intuitively, a subteam as well as its replacement should be similar to the original team in terms of collaboration structure and skills because they work on the same and well-specified task. To address this, we propose \emph{positive team contrasting} to improve team-level representation learning through self-contrasting. The key idea is to minimize the disparity between the embedding of the recommended subteam (i.e., $\mathcal{R}$) and that of the remaining team (i.e., $\mathcal{T}\setminus\mathcal{S}$).

We define the team-level representation (e.g., $\mathcal{R}$) as the weighted aggregation of the members' embeddings:
\begin{equation}
g(\mathcal{R}) = \sum_{t \in \mathcal{R}} w_t\mathbf{Z}[t,:],\qquad \sum_{t \in \mathcal{R}} w_t = 1,
\label{eqn:team _embedding}
\end{equation}
\noindent where $w_t$ denotes the weights for aggregation. In this paper, we set the weight as $w_t=\frac{1}{|\mathcal{R}|}$.

Having obtained the team embeddings from Eq.~\eqref{eqn:team _embedding}, we can now define the \emph{team contrasting loss} as follows:
\begin{equation}
 L_{\mathrm{contra}} = \frac 1 s\sum_{i=1}^sf(g(\mathcal{R}_i), g(\mathcal{T}_i \setminus \mathcal{R}_i )),
\end{equation} 
\noindent where $s$ is the number of sample teams, $\mathcal{R}_i$ denotes the subteam to be replaced from the $i$-th team (i.e., $\mathcal{T}_i$) from the network $\mathcal{G}$. $f(\cdot,\cdot)$ computes the cosine similarity between the two vectors. $L_\mathrm{contra}$ calculates the average similarity between the replaced teams and the remaining teams after members' departure. By minimizing the self-supervised contrastive loss, we enforce the matching among team-level embeddings regarding collaboration structure and skills.

According to the two fundamental principles in designing \emph{subteam replacement} algorithms, i.e., structural and skill match, we further define the following objectives, including (1) \textit{skill loss} for skill similarity between nodes, (2) \textit{structural loss} to evaluate the structural differences between the clustering results and original adjacency matrix, and (3) \textit{clustering loss} to measure the quality of the clustering layer. We detailed the training objective as follows.

\noindent\textit{Skill loss.} From the perspective of skills, the nodes from the same group after the clustering layer are supposed to be similar, therefore, we design the \textit{skill loss} to train the model towards this objective. We first apply a pairwise similarity measure to the attribute matrix (i.e., $\mathbf{X}$) and the cluster assignment matrix (i.e., $\mathbf{C}$) as follows:
\begin{equation}
    \begin{aligned}
     \mathbf{Y}_1 = \mathrm{PairSim}(\mathbf{X},\mathbf{X}), \mathbf{Y}_2 = \mathrm{PairSim}(\mathbf{C},\mathbf{C}),
    \end{aligned}
    \label{eqn:loss_skill_1}
    \end{equation} 
\noindent where $\mathrm{PairSim}(\mathbf{P, Q})=\hat{\mathbf{P}}\hat{\mathbf{Q}}^T$ for $\mathbf{P}, \mathbf{Q}\in\mathbb{R}^{n\times d}$ and $\hat{\mathbf{P}}$, $\hat{\mathbf{Q}}$ denotes the row-normalized $\mathbf{P}$, ${\mathbf{Q}}$, respectively. The matrices, $\mathbf{Y}_1$ and $\mathbf{Y}_2$, represent the similarity between all pairs of nodes regarding node attributes and the cluster assignment results, respectively. We then compare the two matrices from Eq.~\eqref{eqn:loss_skill_1} using PairSim:
\begin{equation}
    \mathbf{Y}_3 = \mathrm{PairSim}(\mathbf{Y}_1,\mathbf{Y}_2),
\end{equation}
\noindent $\mathbf{Y}_3$ encodes the closeness between the two similarity measures on attributes and cluster assignment. Therefore, by aggregating the diagonal elements in $\mathbf{Y}_3$, we can obtain the skill similarity loss from the clustering results as follows,
\begin{equation}
    L_{\mathrm{skill}} = -\mathrm{tr}(\mathbf{Y}_3), 
\end{equation}
\noindent where $\mathrm{tr}(\cdot)$ denotes the trace of the input matrix. Through minimizing $L_\mathrm{skill}$, the clustering layer is optimized to assign nodes with similar skills to the same cluster.

\noindent \textit{Structural loss.} Ideally, nodes that are topologically adjacent are supposed to be clustered together. Inspired by previous studies on supervised node clustering\cite{DiffPool2018}, we exploit the \emph{structure loss} $L_\mathrm{link}$ to encode structural similarity for the cluster assignment results. Specifically, we define the objective as follows:
\begin{equation}
 L_{\mathrm{structural}} = \| \mathbf{A} - \mathbf{C}\mathbf{C}^T\|_F,
\end{equation}
where $\| \cdot \|_F$ denotes the Frobenius norm. 

\noindent \emph{Clustering loss.} Regarding the cluster assignment results, a cluster assignment vector with even values (i.e., close to uniform distribution) represents the assignment of the corresponding node is of high uncertainty. Hence, it is more preferable that the node cluster assignment vector is close to a one-hot vector. We now define the \textit{clustering loss} as:
\begin{equation}
 L_{\mathrm{clustering}} = \frac 1 n \sum_{i=1}^n \mathrm{Entropy}(\mathbf{C}[i,:]),
\end{equation}
where $\mathrm{Entropy}$ denotes the entropy function, which is calculated as 
\begin{equation}
 \mathrm{Entropy}(\mathbf{C}[i,:]) = -\sum_{j=1}^c \mathbf{C}[i,j] \log \mathbf{C}[i,j],
\end{equation}
We combine the aforementioned objectives as follows:
\begin{equation}
    L = L_\mathrm{contra} + b_1 L_\mathrm{skill} + b_2 L_\mathrm{structural} + b_3 L_\mathrm{clustering},
\label{eqn:total_loss}
\end{equation}
where $L$ represents the total loss and $b_i, i \in \{ 1, 2, 3 \}$ are the regularization parameters that balance the scales of training objectives for different datasets. $L$ is applied to guide the training of the entire GNN (i.e., \textit{team network encoder}).

In this way, we have developed a self-supervised model \emph{team network encoder} trained by \emph{positive team contrasting}, a training scheme specially designed for team-related problems. Compared to models trained by node or subgraph (team) labels to learn team representations in supervised ways ~\cite{DBLP}, \emph{team network encoder} has two advantages:
\begin{itemize}
    \item Labels have limited expressions in \emph{subteam replacement} because they do not contain critical team information. \emph{Team network encoder} instead utilizes existing team samples to encode the knowledge of collaboration structure and skill in member- and team-level representations.
    \item \emph{Positive team contrasting} utilizes data more efficiently. Unlike supervised training schema where datasets are clearly divided into training sets and testing sets with the former being much less convincing for validations, the training team samples for \emph{team network encoder} can still serve as valid testing cases with different subteams to be replaced. 
\end{itemize}
\noindent\textbf{Subteam Recommender.}
Having the proposed \emph{team network encoder} model, we are equipped with effective representations of teams and their members. 
Given an original team $\mathcal{T}$ and the subteam to be replaced $\mathcal{R}\subset\mathcal{T}$, we propose a \emph{subteam recommender} that performs an efficient within-cluster search based on subteam representation (Algorithm \ref{alg:inference}) where $g$ is the predefined function (Eq.~\eqref{eqn:team _embedding}) for calculating subteam embeddings, $f$ computes the cosine similarity. $\mathrm{Set}$ is a deduplication function for containers and $\mathrm{Product}$ generates Cartesian products on the assigned clusters.

\textit{Subteam recommender} generates candidates only from clusters where subteam members to be replaced are assigned to. Compared to generating candidates across the entire dataset, this pruning technique significantly improves the efficiency of the algorithm (Lemma 1). Additionally, the proposed \genius~framework can guarantee that the optimal solution is included in the candidates since nodes with similar embeddings will be clustered together and thus, it is effective to prune unqualified candidates through clustering.
\begin{lemma}
Within-cluster search is $O(c^{\lvert \mathcal{R} \rvert})$ more efficient than searching over the entire dataset.
\end{lemma}
\begin{proof}
With the calculation of each iteration viewed as a constant, searching over the entire dataset has a time complexity of $O(n^{\lvert \mathcal{R} \rvert})$ because there are $n^{\lvert \mathcal{R} \rvert}$ combinations of new subteams (repetitions included) generated from the entire dataset. On average there are $n/c$ nodes within a cluster and thus, the time complexity for within-cluster search is in $O((n/c)^{\lvert \mathcal{R} \rvert})$. Within-cluster search has time complexity $O(1 /{c^{\lvert \mathcal{R} \rvert}})$ of that of searching over the entire dataset and thus is $O(c^{\lvert \mathcal{R} \rvert})$ more efficient.
\end{proof}
\begin{lemma}
The search algorithm is guaranteed to produce $\mathcal{S}$ where $\lvert \mathcal{S} \rvert \leq \lvert \mathcal{R} \rvert$.
\end{lemma}
\begin{proof}
$\mathrm{Product}$ performs cartesian product on lists, generating products of size $\lvert \mathcal{R} \rvert$ (i.e.,
$
    \lvert \mathcal{D} \rvert = \lvert \mathcal{R} \rvert
$).
If several subteam members belong to the same cluster, some products contain repeating elements. In step $6$, repeating elements will be eliminated by the deduplication operation $\mathrm{Set}$ (i.e.,
$
    \lvert \mathcal{E} \rvert \leq \lvert \mathcal{D} \rvert = \lvert \mathcal{R} \rvert
$).
Since $\mathcal{S}$ is generated from the set of all $\mathcal{E}$'s, $\lvert \mathcal{S} \rvert \leq \lvert \mathcal{R} \rvert$ holds.
\end{proof}
It is also worth stressing that with the entire dataset fed into \emph{team network encoder} once, all teams and their subteams to be replaced can be directly passed into \emph{subteam recommender}, making the amortized complexity of \textsc{Genius} very low.
\begin{algorithm}[!t]
\caption{Subteam Recommender}\label{alg:inference}
\begin{algorithmic}[1]
\Require (1) the original team $\mathcal{T}$; (2) a subteam to be replaced $\mathcal{R}$; (3) the cluster hard assignment vector $\mathbf{h}$; (4) node embedding matrix $\mathbf{Z}$; (5) remaining social network $\mathcal{M}$; (6) clusters $\mathcal{K}$
\Ensure The optimal subteam $\mathcal{S} \in \mathcal{M}$ where $\lvert \mathcal{S} \rvert \leq \lvert \mathcal{R} \rvert$
\State Initialize similarity $y$ = 0, $\mathcal{S}$ = \{\}
\State Pre-compute remaining subteam embedding $\mathbf{r}$ $\gets$ $g(\mathcal{T} \setminus \mathcal{R}$)
\State Compute clusters where $\mathcal{R}$ are assigned $\mathbf{c}$ $\gets$ $\mathbf{h}[\mathcal{R}]$
\State Generate candidate new subteams $\mathcal{Q}$  $\gets$ $\mathrm{Product}$($\mathcal{K}[\mathbf{c}]$)

\For{each candidate $\mathcal{D}$ \textbf{in} $\mathcal{Q}$}
\State Remove members in $\mathcal{T}$ from candidate $\mathcal{E}$ $\gets$ $\mathcal{D}$ $\setminus \mathcal{T}$
\State Deduplicate candidate $\mathcal{E}$ $\gets$ $\mathrm{Set}$($\mathcal{E}$)
\State Compute embedding for the new subteam $\mathbf{n}$ $\gets$  $g(\mathcal{E})$
\State Compute similarity $t$ $\gets$ $f(\mathbf{r}, \mathbf{n}$)
\If{$t$ \textgreater  $y$}
    \State Update $\mathcal{S}$ $\gets$ $\mathcal{D}$, $y \gets t$
\EndIf
\EndFor
\end{algorithmic}
\end{algorithm}
\subsection{\textbf{Analysis}}\label{sec:analysis}
\vspace{-1mm}
Empirically, we observe that as the number of node features increases, graph kernel-based methods fail to generate high-quality replacements. In this sub-section, we will first analyze the limitations of graph kernel-based methods and then corroborate how the proposed \genius ~addresses them.

\noindent\textbf{Effectiveness.}
As the number of node features increases, each feature dimension no longer represents a general class (i.e., coarse-grained features) where little correlation exists among different dimensions. Instead of being limited by a fixed set of genres, nodes with more specialized descriptions are becoming dominant these days, which remarkably increases the difficulty for graph kernel-based methods to capture the potential correlations among various features.\\
To be specific, given two labelled graphs $\mathbf{G}_i = \{ \mathbf{A}_i, \mathbf{L}_i \}, i = 1,2,$
\cite{randomwalk_nips,randomwalk_tong,Liangyue2015} is computed as:

\begin{equation}
    \mathrm{Ker}(\mathbf{G}_1, \mathbf{G}_2) = \mathbf{y}(\mathbf{I} - a\mathbf{A}_\times)^{-1} \mathbf{L}_\times \mathbf{x},
\label{kernel}
\end{equation}
where $\mathbf{A}_\times = \mathbf{L}_\times (\mathbf{A}_1 \otimes \mathbf{A}_2)$, $\otimes$ representing the Kronecker product of two matrices, $\mathbf{x}$ and $\mathbf{y}$ represent starting and stopping probability vectors for random walk, which are usually defined as uniform, and $a$ is a decay factor. $\mathbf{L}_\times$ is a diagonal matrix calculated as:
\begin{equation}
\mathbf{L}_\times = \sum_{i=1}^d \mathrm{diag}(\mathbf{L}_1(:,i)) \otimes \mathrm{diag}(\mathbf{L}_2(:,i)),
\label{Lx}
\end{equation}
where $\mathrm{diag}$ represents the diagonalization operation. It represents matches of labels among nodes

Random walk graph kernel bases on paths. The $\mathbf{L}_\times$ term in the expression eliminates paths between nodes where labels are different, regardless of the extent they vary. Most importantly, random graph kernel fails to recognize correlation among different features. Although some previous study includes skill pair matrix in the algorithm \cite{Zhaoheng2021}, the essence of those calculations is still a rough match of skill vectors. For instance,
\emph{pattern recognition}, \emph{feature extraction} and \emph{security} are three representative node features in DBLP dataset. When it comes to team member substitution, \emph{pattern recognition} and \emph{feature extraction} have strong correlation with one another. However, by the nature of Kronecker product, an author with publication merely in \emph{pattern recognition} is simply viewed as unmatched with another author with publication only in \emph{feature extraction}. This does not only confine potential applications in subteam member substitutions, but also gives rise to inaccurate replacement since \emph{pattern recognition} is considered to be as closely correlated with \emph{security} as \emph{feature extraction}.

\textsc{Genius}, on the other hand, is capable of handling graphs with massive node features. By including sample teams in training, the model learns to identify correlations between node features in teams study and make optimal replacement choices.

\noindent\textbf{Efficiency.}
Given $d$ node features, the optimal time complexity of graph kernel-based methods is $O(d^2)$ \cite{Zhaoheng2021} while the time complexity of \genius~ is $O(d)$. To be specific, the number of parameters to train and the dimension of vector multiplication in \emph{subteam recommender} are both $O(d)$ while the remaining calculations are not affected by the variance of $d$.
Besides, $O(n)$ cases can be passed into \emph{subteam recommender} with the entire dataset encoded by \textit{team network encoder} only once; therefore the amortized time complexity for \textsc{Genius} is lower.

To briefly conclude, we provide a theoretical proof that as the number of node features increases (i.e., finer grained features), \genius ~can present a more stable performance in terms of efficiency compared to graph kernel-based methods. We also demonstrate our analysis through extensive empirical evaluations and the details are in Sec.~\ref{sec:exp}


%% file: 05experiment.tex
\section{Experiments}\label{sec:exp}
In this section, we perform empirical evaluations to demonstrate the effectiveness and efficiency of \textsc{Genius}. Specifically, we aim to answer the following research questions:

\begin{itemize}
    \item \textbf{RQ1.} How effective is the proposed framework \textsc{Genius} in generating replacement of unavailble subteams? 
    How does performance change as the size of the skill set (i.e., number of graph node features) varies?
    \item \textbf{RQ2.} How much will the performance of \textsc{Genius} change by providing different percentages of subteam members to be replaced?
    \item \textbf{RQ3.} How efficient is the proposed framework for \emph{subteam replacement}? 
\end{itemize}
\subsection{\textbf{Experimental Setup}}\label{subsec:exp_setup}

\noindent\textbf{Evaluation Datasets.}
\begin{table}[!t]
	\centering
	  \caption{Statistics of evaluation datasets.}
	{\begin{tabular}{lccc}
    \toprule
    \textbf{Datasets} & DBLP & IMDB\\
    \midrule
    \# nodes & $26,351$ & $13,472$ \\
    \# edges & $54,044$ & $665,166$\\
    \# features & $7,551$ & $2,040$\\
    \# teams & $7,052$ & $818$\\
    \bottomrule
	\end{tabular}
	}
  
	\label{tab:data_stats}
\end{table}
In the experiment, we adopt two real-world datasets widely used in previous research\cite{Liangyue2015,Zhaoheng2021} but with around $\times 150$ more abundant node features. Table \ref{tab:data_stats} summarizes the statistics of each dataset. Detailed descriptions are as follows:

\begin{itemize}
    \item \textbf{DBLP}\footnote{http://arnetminer.org/citation} provides computer science bibliographic information. We build a graph where nodes represent authors, edges represent co-authorship and node features are each author's paper keywords. Teams refer to the authors of each paper.
    
    \item \textbf{IMDB}\footnote{https://grouplens.org/datasets/hetrec-2011/} contains statistics of actors/actresses and movies where nodes represent actors/actresses, edges represent the number of movies where actors/actresses co-starred and node features are tags of movies in which actors/actresses played. Teams are represented by crews of actors/actresses in a movie. 
    
\end{itemize}
\smallskip
\noindent\textbf{Comparison Methods.} We compare \textsc{Genius} with two categories of baseline methods for solving \textit{subteam replacement}, including (1) \emph{graph kernel}-based method, which prior investigations take as optimal
solution, and (2) \emph{supervised encoder}-based method, which the capability of \textsc{Genius} to learn team- and member- level representations is compared to. 
\begin{itemize}
    \item \textbf{Graph Kernel}~\cite{Liangyue2017,Zhaoheng2021} based method applies random walk graph kernel to find the optimal result. The entire social network is searched to obtain all possible combinations of size $\lvert \mathcal{R} \rvert$ as new subteam candidates and the optimal solution is the subteam that produces the highest graph kernels w.r.t. original team.
    \item \textbf{Supervised Encoder}~\cite{DBLP} based method is of similar GNN-backboned structure as \emph{team network encoder} in \textsc{Genius} but trained by graph node labels. 
    For DBLP, label for each node represents the research field where the author has maximal publications; for IMDB, label for each node represents the movie genre where the actor stars most \cite{Liangyue2015,Zhaoheng2021}.
    \emph{Supervised encoder} generates node embeddings along with label assignments. 
    The entire network is searched to obtain all possible combinations of size $\lvert \mathcal{R} \rvert$ as new subteam candidates and the optimal solution is the subteam that produces subteam embedding closest to the remaining subteam embedding (using the same subteam embedding function $g$ as \textsc{Genius}).
\end{itemize}
\smallskip
\noindent\textbf{Training Methods.} For each dataset, $60\%$ of the teams are selected as training samples, $20\%$ are for validation while the remaining $20\%$ (labelled as $\mathcal{Z}, \mathcal{Z} \subset \mathcal{G}$) are for testing as (both for \textsc{Genius} and for baselines). For balancing parameters in Eq.~\eqref{eqn:total_loss}, $b_1, b_2, b_3$ are set as $1, 100, 1$ respectively for DBLP and $100, 100, 10$ respectively for IMDB.\\

\noindent\textbf{Evaluation Metrics.}
For quantitative analysis, there is no sole golden criteria to evaluate the effectiveness of the results. Thus, we study the disparity between the optimized new teams and original teams based on multiple graph similarity metrics to make comparisons more convincing. 
In this paper, the following metrics, all applied in the versions where both structural properties and feature information are considered, are applied to conduct a comprehensive evaluation on the performance of different methods.\\
\begin{itemize}
    \item \textbf{Graph Edit Distance (GED)} \cite{GED} is defined as the minimum cost of editing path to transform one graph to another. For two identical graphs, the Graph Edit Distance (GED) is zero. Specifically, for each testing team $t_0 \in \mathcal{Z}$, $\mathrm{GED}(\cdot)$ measures the graph edit distance based disparity between the original team $t_0$ and the new team $t_1$ \cite{networkx}. 
    Graph edit distance for each method in each dataset is the average of all testing samples $t_0 \in \mathcal{Z}$.
    \item \textbf{Shortest Path Graph Kernel} \cite{ShortestPath} is a path-based graph kernel where graphs are decomposed into shortest paths and compared in pairs by length and end-point labels (i.e., node features) \cite{graphkit-learn,Grakel}. Specifically, for each testing team  $t_0 \in \mathcal{Z}$, we calculate the disparity $\mathrm{D}_1(t_0,t_1)$ between the original team $t_0$ and the new team $t_1$ based on shortest path graph kernel function $\mathrm{SP}(\cdot)$.
    \begin{align*}
        \mathrm{D}_1(t_0,t_1) =  \mathrm{abs}(\mathrm{SP}(t_0, t_1),\mathrm{SP}(t_0, t_0))/\mathrm{SP}(t_0, t_0)
    \end{align*}
        where $\mathrm{abs}(a,b)$ represents the absolute difference of $a$ and $b$.
The disparity based on shortest path graph kernel for each method in each dataset is the average of all testing samples $t_0 \in \mathcal{Z}$.
    \item \textbf{Marginalized Graph Kernel} ~\cite{Marginalized,MarginalizedExtension} is a walk-based kernel that takes the inner product of walks on the two graphs. \cite{graphkit-learn} Specifically, for each testing team  $t_0 \in \mathcal{Z}$, we calculate the disparity $\mathrm{D}_2(t_0,t_1)$ between the original team $t_0$ and new team $t_1$ based on marginalized graph kernel function $\mathrm{M}(\cdot)$.
\begin{align*}
        \mathrm{D}_2(t_0,t_1) =  \mathrm{abs}(\mathrm{M}(t_0, t_1),\mathrm{M}(t_0, t_0))/\mathrm{M}(t_0, t_0)
\end{align*}
        where $\mathrm{abs}(a,b)$ represents the absolute difference of $a$ and $b$.
The disparity based on marginalized graph kernel for each method in each dataset is the average of all testing samples $t_0 \in \mathcal{Z}$.
\end{itemize}

Besides quantitative analysis, we also include qualitative analysis to further demonstrate that \textsc{Genius} generates promising results in real-world applications. For qualitative analysis, we perform case studies both on the plausibility of cluster assignments and final replacement results.

All experiments are run on a Tesla V100 GPU.
\subsection{\textbf{Effectiveness Results(RQ1)}}
\noindent\textbf{Quantitative Analysis}

\begin{figure*}[t!]
\centering
    \begin{subfigure}[t]{0.32\linewidth}
           \centering
           \includegraphics[width=\linewidth]{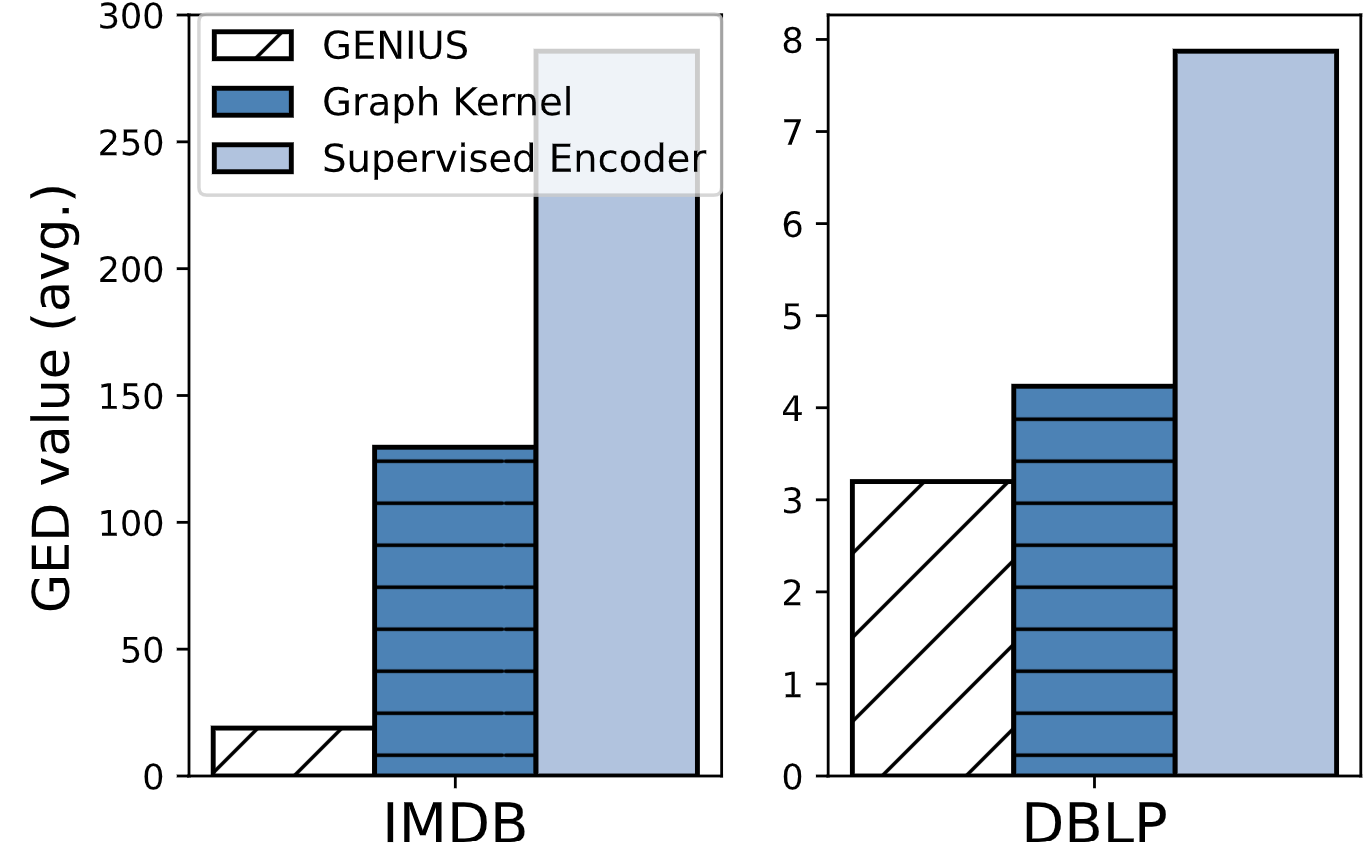}
            \captionsetup{font=scriptsize}
            \caption{Disparity based on graph edit distance (GED).}
            \label{overall comparison:GED}
    \end{subfigure}
    \begin{subfigure}[t]{0.32\linewidth}
            \centering
            \includegraphics[width=\linewidth]{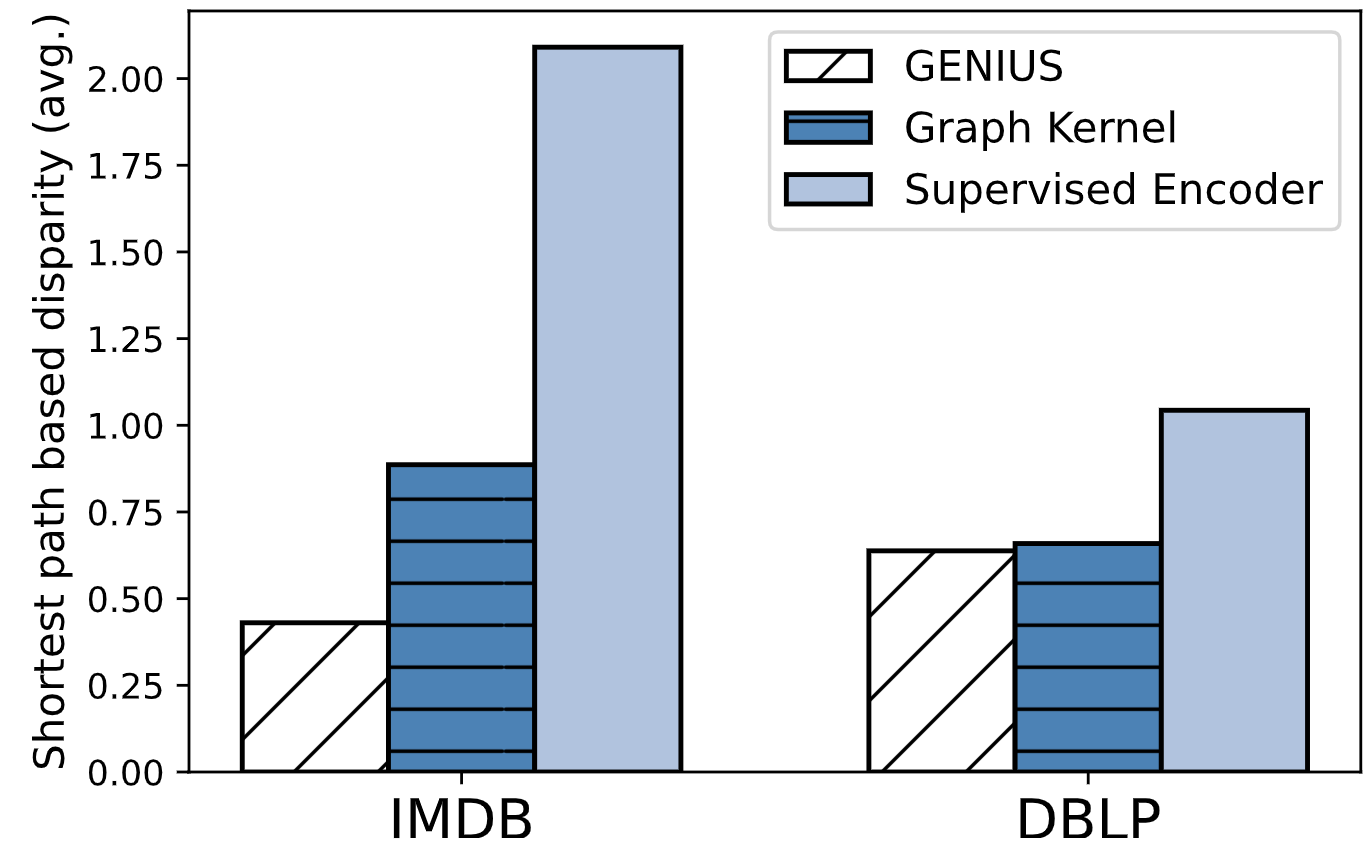}
            \captionsetup{font=scriptsize}
            \caption{Disparity based on shortest path graph kernel.}
            \label{fig:b}
    \end{subfigure}
    \begin{subfigure}[t]{0.32\linewidth}
            \centering
            \includegraphics[width=\linewidth]{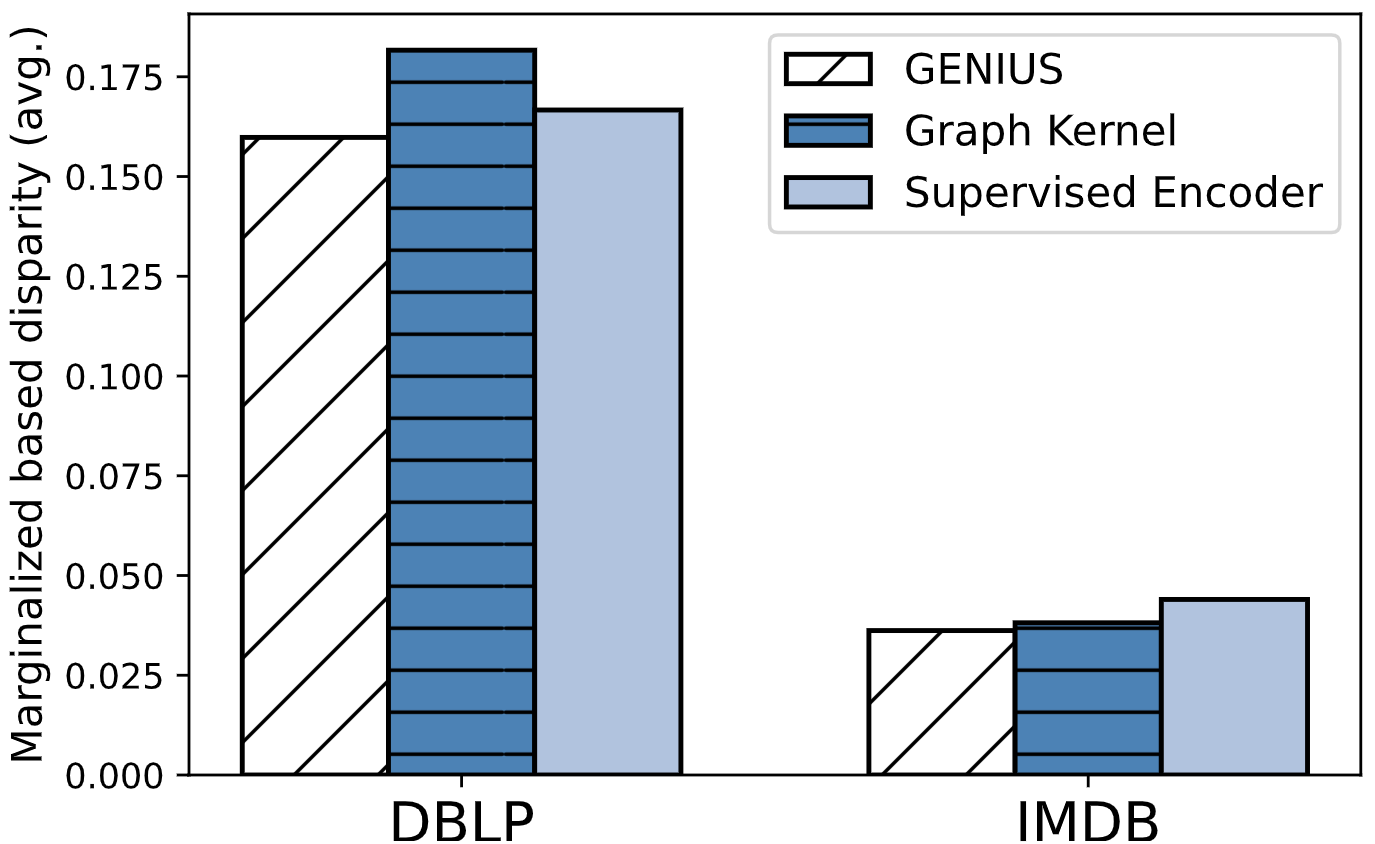}
            \captionsetup{font=scriptsize}
           \caption{Disparity based on marginalized graph kernel.}
            \label{fig:b}
    \end{subfigure}
    \caption{Performance evaluation w.r.t. different datasets.}
    \label{overall comparison}
\end{figure*}

\emph{Overall Comparison.} We first conduct experiments to examine the general performance of \textsc{Genius} and compare the results w.r.t. baselines. For each original team $t_0 \in \mathcal{Z}$, we randomly select different percentages of members as the subteam to be replaced. Both \textsc{Genius} and the two baselines are tested with the same set of test cases. The results are visualized in Figure ~\ref{overall comparison}. 
Accordingly, we have the following observations, including: \textbf{(1)} Based on various graph similarity metrics and across different datasets, \textsc{Genius} generates results with significantly lower graph disparity w.r.t. original teams: disparity based on marginalized graph kernel for \textsc{Genius} in DBLP drops by $15 \%$ from that for \emph{graph kernel}-based method; disparity based on shortest path graph kernel for \textsc{Genius} in IMDB is only $1/5$ of that for \emph{supervised encoder}-based method; GED based graph disparity for \textsc{Genius} in IMDB is only $1/15$ of that for \emph{supervised encoder}-based method and $1/7$ of that for \emph{graph kernel}-based method. This indicates that compared to baselines, \textsc{Genius} is more capable of learning effective team-level collaboration structure and skill representations. \textbf{(2)} The graph disparity of baseline models increases drastically when the average team size increases: for instance, on IMDB (average team size is around 17), both baselines generate results with average GED rising to around $\times 100$ more than that of DBLP (average team size is around 3) while \textsc{Genius} has a steady performance (Figure \ref{overall comparison:GED}). From this, we can view that \textsc{Genius} generates more stable team representations than baselines. \textbf{(3)} \emph{Supervised encoder}-based method has the worst performance of the three, confirming our previous assertion that training with labels is not appropriate for team-related tasks. 
According to overall comparisons, \textsc{Genius} outperforms baselines to a great extent.

\begin{figure*}[t!]
\centering
    \begin{subfigure}[t]{0.32\linewidth}
           \centering
           \includegraphics[width=\linewidth]{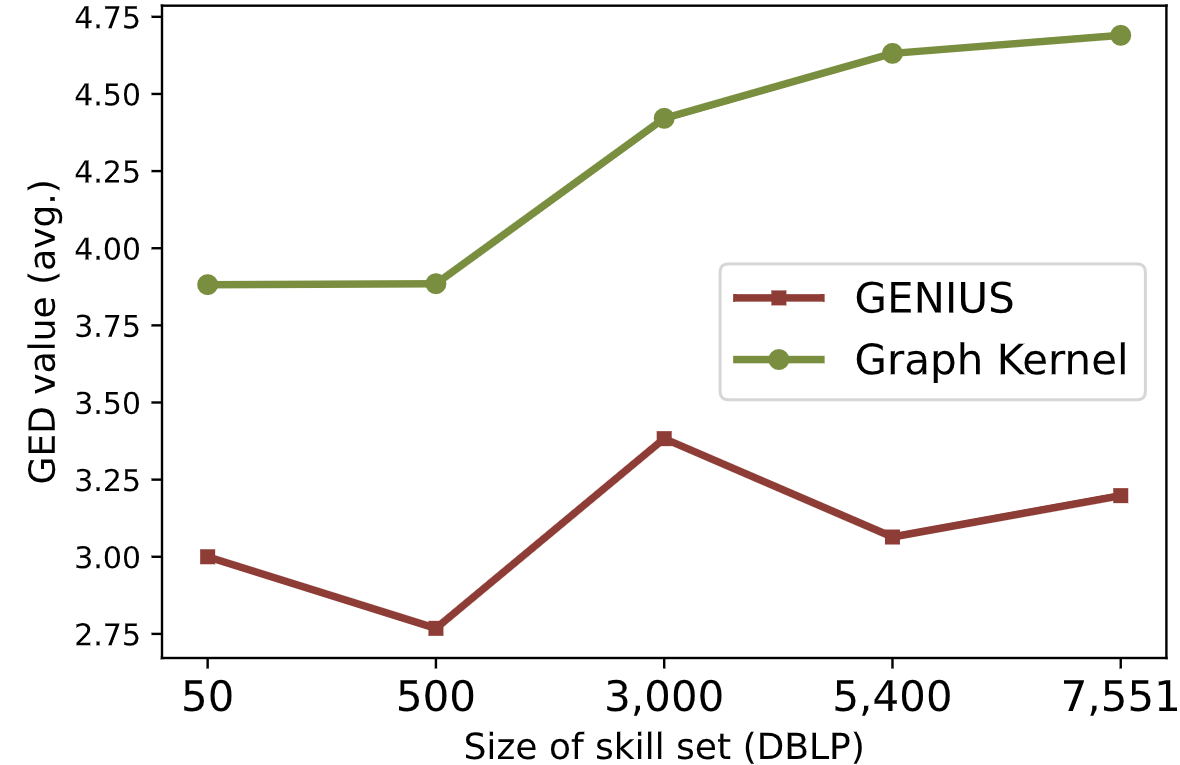}
            \captionsetup{font=scriptsize}
            \caption{Disparity based on graph edit distance (GED).}
            \label{fig:Case1}
    \end{subfigure}
    \begin{subfigure}[t]{0.32\linewidth}
            \centering
            \includegraphics[width=\linewidth]{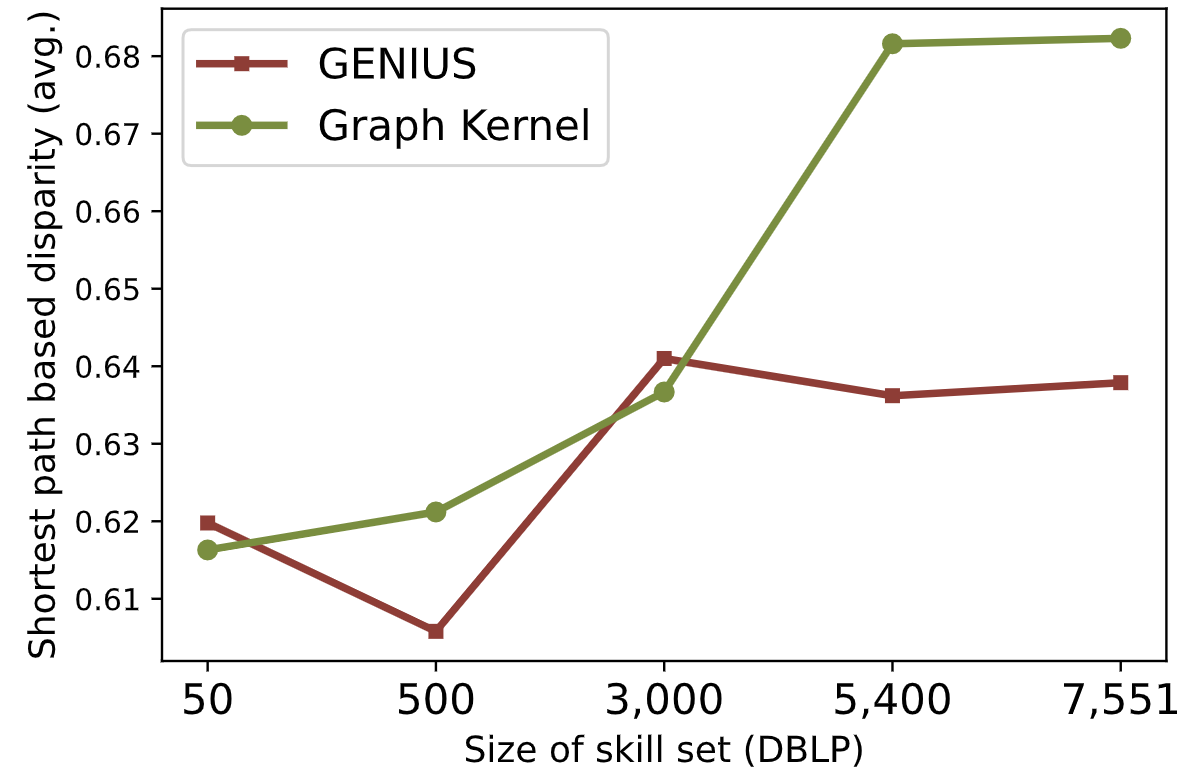}
            \captionsetup{font=scriptsize}
            \caption{Disparity based on shortest path graph kernel.}
            \label{fig:b}
    \end{subfigure}
    \begin{subfigure}[t]{0.32\linewidth}
            \centering
            \includegraphics[width=\linewidth]{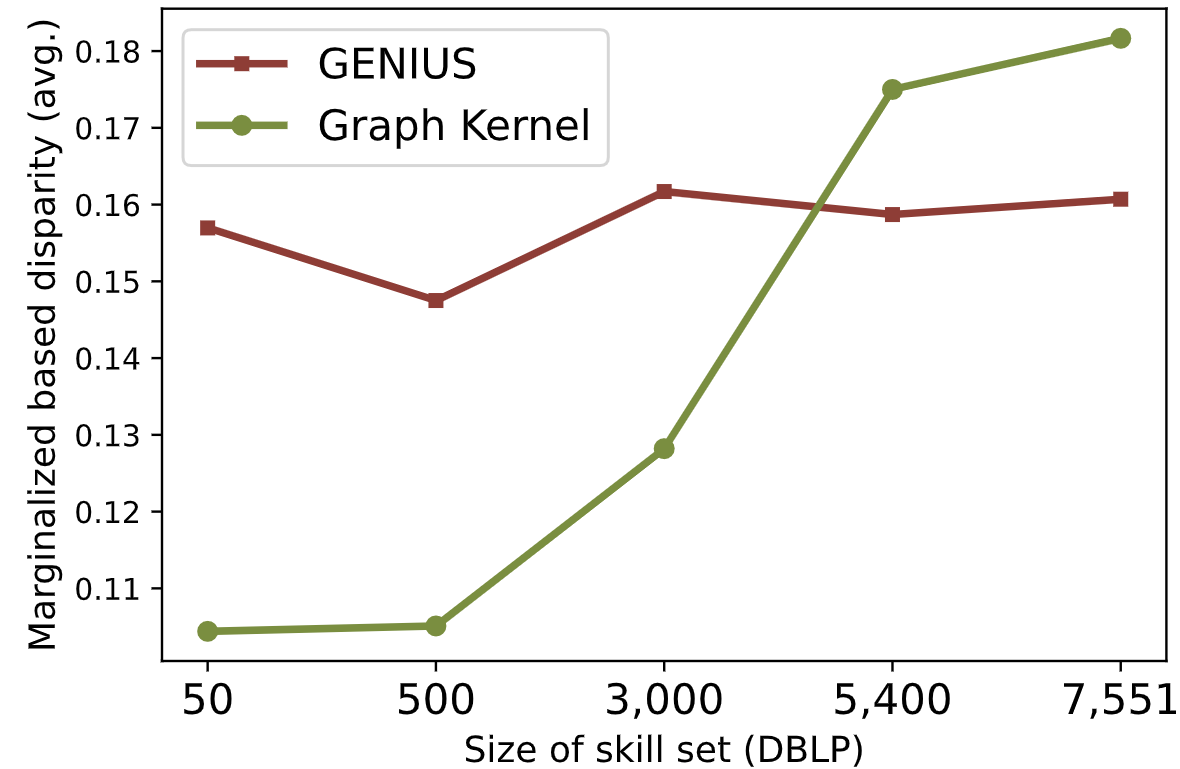}
            \captionsetup{font=scriptsize}
            \caption{Disparity based on marginalized graph kernel.}
            \label{fig:b}
    \end{subfigure}
    \caption{Performance variation w.r.t. size of skill set (DBLP).
    }
    \label{features}
\end{figure*}
\emph{Performance as the size of skill set varies.}
This part addresses our discussion on the limitation of graph kernels as the skill set scales up in section ~\ref{sec:analysis}.
We pick DBLP as the evaluation dataset because it allows a wider range of node feature variation. The same experiment settings are applied except that different numbers of node features are randomly selected from the original dataset, ranging from $50(1\%)$ to $7,551(100\%)$. Both \textsc{Genius} and \emph{graph kernel}-baseline are tested with the same set of skills and test cases. Figure ~\ref{features} visualizes the results where as the size of skill set increases, the disparity between the optimized new teams and original teams increases drastically for \emph{graph kernel}-based approach while \textsc{Genius} has a stable performance. The results confirm our previous statement that the performance of \emph{graph kernel}-based method is sabotaged by its inability to recognize potential correlations between various skills.

\noindent\textbf{Qualitative Analysis}

\emph{Cluster Assignment}. We extract the cluster where \emph{Jiawei Han} belongs and it turns out that \emph{Philip S. Yu}, \emph{Yizhou Sun}, \emph{Xifeng Yan} also belong to the same cluster. This is consistent with intuitions because all of them share very similar research interests and connections with \emph{Han}. If \emph{Han} becomes unavailable in a team, the most appropriate candidates are included within the cluster that he belongs to, confirming the effectiveness of cluster assignments generated by \emph{team network encoder} and the within-cluster search performed by \emph{subteam recommender}.

\emph{Replacement Results}. In a team of $7$, we assume a subteam of size $2$, composed of \emph{Vipin Kumar} and \emph{Huzefa Rangwala}, becomes unavailable. \textsc{Genius} generates \emph{Rakesh Agrawal} as the optimal result. This is plausible because \emph{Rakesh Agrawal} has research interests in data mining and experience that highly match with \emph{Kumar} and \emph{Rangwala}, meaning that the new team is able to perform similar tasks. The newly-formed team of $6$ is also expected to cope well because \emph{Rakesh Agrawal} has connections with the remaining $5$ members.

\begin{figure*}[t!]
\centering
    \begin{subfigure}[t]{0.32\linewidth}
           \centering
           \includegraphics[width=\linewidth]{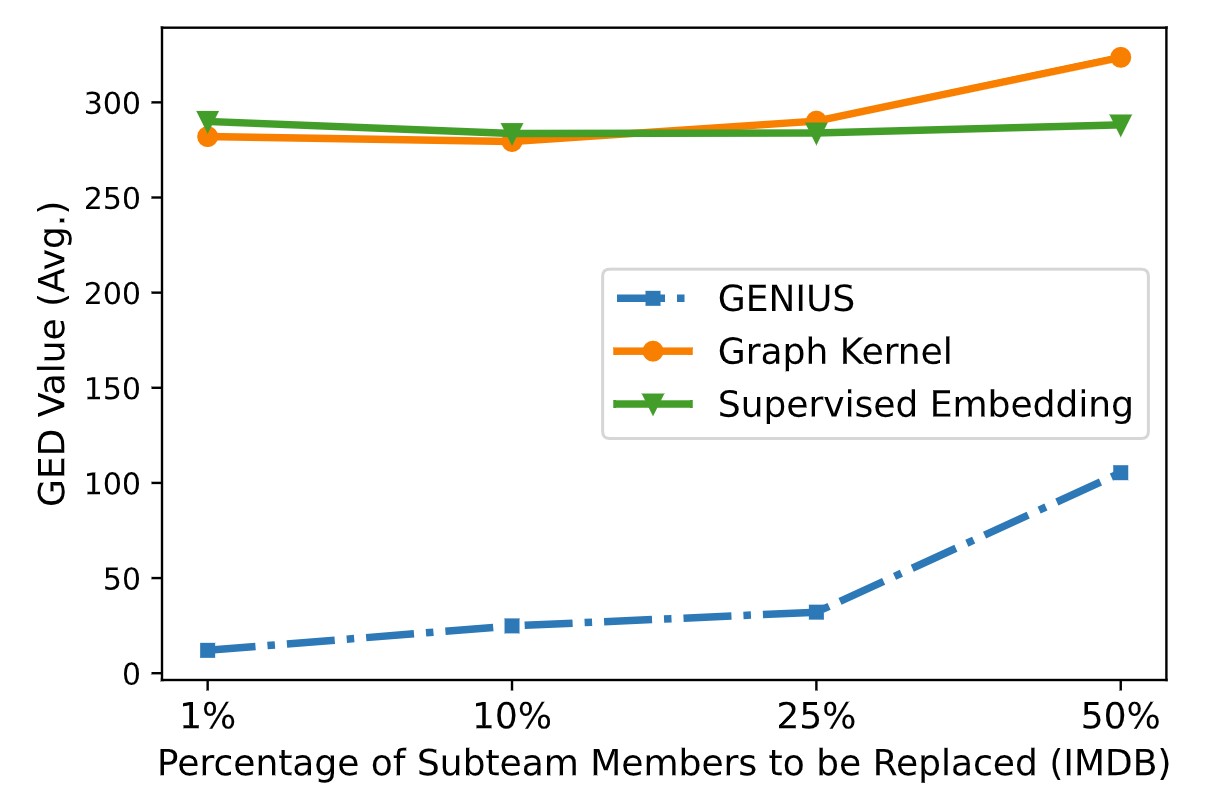}
            \captionsetup{font=scriptsize}
            \caption{Disparity based on graph edit distance (GED).}
            \label{fig:Case1}
    \end{subfigure}
    \begin{subfigure}[t]{0.32\linewidth}
            \centering
            \includegraphics[width=\linewidth]{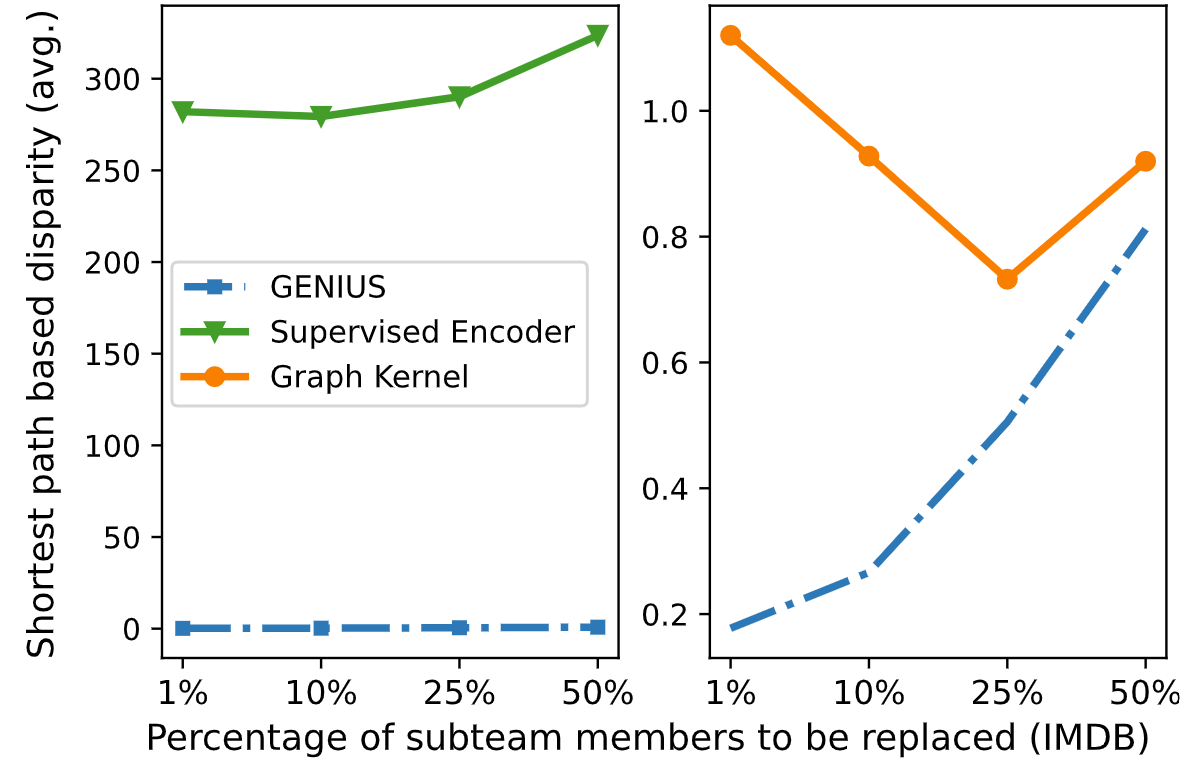}
            \captionsetup{font=scriptsize}
            \caption{Disparity based on shortest path graph kernel.}
            \label{fig:b}
    \end{subfigure}
    \begin{subfigure}[t]{0.32\linewidth}
            \centering
            \includegraphics[width=\linewidth]{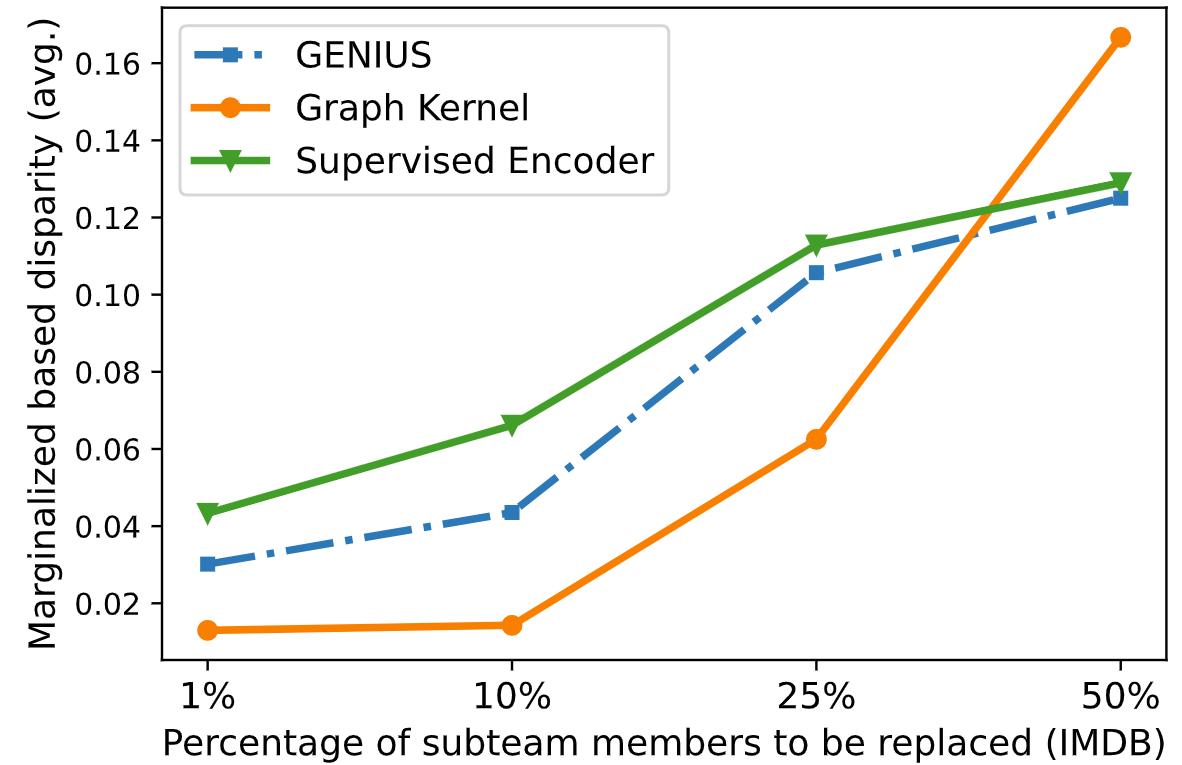}
            \captionsetup{font=scriptsize}
            \caption{Disparity based on marginalized graph kernel.}
            \label{fig:b}
    \end{subfigure}
    \caption{Performance variation w.r.t. percentage of subteam members to be replaced (IMDB).}
    \label{teamsize}
\end{figure*}
\subsection{\textbf{Hyper-Parameter Study (RQ2)}}
In this section, we mainly study how varying percentages of subteam members to be replaced affect the quality of results. We pick IMDB as the evaluation dataset because it includes teams with sizes significantly larger (around 17) than that of DBLP (around 3) and thus allows a wider range of subteam member percentage studies. \textsc{Genius} and two baselines are tested with the same original teams and subteams to be replaced. The resulting graph disparity is visualized in Figure~\ref{teamsize}. We can observe that as the percentage of subteams to be replaced increases, the graph disparity of all three methods increases, which is consistent with our intuition because fewer percentage of new team members are identical to that of the original team. However, we can see that regardless of the increases in graph disparity, what we observe in \emph{Overall Comparison} stands for every subteam percentage: \textbf{(1)} For every percentage of subteam members to be replaced, \textsc{Genius} has graph disparity that is significantly lower than the two baselines. This indicates that regardless of team sizes, \textsc{Genius} is more capable of mining critical knowledge to construct a team than baselines. \textbf{(2)} \emph{Supervised encoder}-based method remains to have the highest graph disparity regardless of the percentages of subteam members to be replaced.
\subsection{\textbf{Efficiency Results (RQ3)}}\label{subsec:efficiency}
For models that require training (\textsc{Genius} and \emph{supervised encoder}-based method), we take total running time as the sum of training time (total training time for $1,000$ epochs divided by the number of test samples, i.e., $\lvert \mathcal{Z} \rvert$) and inference time. For the method that does not require training (\emph{graph kernel}-based method), the total running time is taken as the time required to obtain optimized solutions.
We investigate the efficiency results from the following aspects:
\begin{itemize}
    \item The overall efficiency differences between \textsc{Genius} and baselines (Table \ref{general_efficiency}).
    It is obvious that \textsc{Genius} performs much faster than the two baselines even if training time is taken into consideration due to its low amortized time complexity. Improvements brought by within-cluster search algorithm can be clearly viewed from the inference time contained in brackets: since both baselines have to iterate over the entire dataset to generate candidates while \textsc{Genius} prunes unqualified candidates by searching only within clusters, the inference time of \textsc{Genius} for an optimized solution is significantly shorter. \emph{Graph kernel}-based baseline takes the longest time to find a solution because it performs more complicated calculations on each subteam member to be replaced. Specifically, \textsc{Genius} has around $\times 600$ higher efficiency than \emph{graph kernel}-based method on DBLP and around $\times 400$ higher on IMDB; \textsc{Genius} has around $\times 3$ higher efficiency than \emph{supervised encoder}-based method on DBLP and around $\times 20$ higher efficiency on IMDB.
    \item Training time differences of \textsc{Genius} as the size of skill set varies (Table \ref{training_time}). From the results, we can observe tiny fluctuations of the average training time per epoch, i.e., $3.12 \pm 0.29$, across different numbers of node features (coefficient of variation $c_v \approx 9\%$). Thus, although variations in the number of node features give rise to different numbers of parameters to be trained, the actual training time per epoch varies within a small range.
    \item Efficiency differences of \textsc{Genius} comparing to \emph{graph kernel} baseline as the size of skill set varies (Table \ref{feature_efficiency}). We can observe that as the number of node features increases, \emph{graph kernel}-based method takes a drastically increasing amount of time to find a solution. This results from the $\mathbf{L}_\times$ term (Eq.~\eqref{Lx}) in graph kernel expression (Eq.~\eqref{kernel}) calculated by iterating over node features, giving an increase in time complexity of $O(d)$. Meanwhile, the average inference time to find a solution using \textsc{Genius} drops as the number of features increases (since the training time is close across various percentages, the total time also drops). This is because with more detailed skill information, nodes become more scattered and thus, within-cluster search generates fewer candidates. For every size of skill set, \textsc{Genius} finds a solution significantly faster than \emph{graph kernel}-based baseline.
    \item Efficiency differences as the percentage of subteam members to be replaced varies (Table \ref{portion_efficiency}). For \textsc{Genius} and \emph{supervised encoder}-based method, the training time for each method is the same across different subteam percentages because the inferences are all performed on the same encoded team network.
    As the percentage of subteam members to be replaced increases, the number of candidates increases and thus, the inference time it takes to obtain a solution increases for both \textsc{Genius} and the two baselines. Regardless of different percentages of subteam members to be replaced, \textsc{Genius} takes significantly less time to find a solution than the two baselines.
\end{itemize}
\begin{table}[!t]
	\centering
	  \caption{Average time to find a solution in seconds}
	  \begin{threeparttable}
	{\begin{tabular}{lcc}
    \toprule
    \textbf{Methods} & \textbf{DBLP}  & \textbf{IMDB}\\
    \midrule
    \textbf{\textsc{Genius}} & 2.27 (0.04) & 11.74 (0.13)\\
    \textbf{Graph Kernel} & 1335.20 & 4432.01\\
    \textbf{Supervised Embedding} & 7.32 (7.21) & 216.36 (215.87)\\
    \bottomrule
	\end{tabular}
	}
 \begin{tablenotes}
        \footnotesize
        \item[*] Numbers in brackets (if any) represent inference time
      \end{tablenotes}
    \end{threeparttable}
          \vspace{2mm}
  
	\label{general_efficiency}
\end{table}

\begin{table}[!t]
	\centering
	  \caption{Average training time for \textsc{Genius} per epoch in seconds w.r.t. size of skill set (DBLP)}
	 \resizebox{\linewidth}{!}{
	{\begin{tabular}{lccccc}
    \toprule
    \textbf{$\#$ Node Features} & $50$  & $500$ & $3,000$ & $5,400$ & $7,551$\\
    \midrule
    \textbf{Average Training Time} & 3.59 & 2.77 & 3.32 & 2.93 & 3.00\\
    \bottomrule
	\end{tabular}
	}
	}
  
	\label{training_time}
\end{table}

\begin{table}[!t]
	\centering
	  \caption{Average time to find a solution in seconds w.r.t. size of skill set (DBLP)}
	  \begin{threeparttable}
	  	 \resizebox{\linewidth}{!}{
	{\begin{tabular}{lccccc}
    \toprule
    \textbf{Methods} & $50$  & $500$ & $3,000$ & $5,400$ & $7,551$\\
    \midrule
    \textbf{\textsc{Genius}} & 2.77 (0.23) & 2.12 (0.15) & 2.41 (0.05) & 2.10 (0.03) & 2.17 (0.04)\\
    \textbf{Graph Kernel} & 12.40 & 96.10 & 570.12 & 924.34 & 1335.20\\
    \bottomrule
	\end{tabular}
	}
	}
	 \begin{tablenotes}
        \footnotesize
        \item[*] Numbers in brackets (if any) represent inference time
      \end{tablenotes}
      \vspace{2mm}
    \end{threeparttable}
  
	\label{feature_efficiency}
\end{table}

\begin{table}[!t]
	\centering
	  \caption{Average time to find a solution in seconds w.r.t. percentages of subteam members to be replaced (IMDB)}
	  \begin{threeparttable}
	\resizebox{\linewidth}{!}{
	{\begin{tabular}{lcccc}
    \toprule
    \textbf{Methods} & $1\%$  & $10\%$ & $25\%$ & $50\%$\\
    \midrule
    \textbf{\textsc{Genius}} & 11.61 (0.01) & 15.30 (3.69)  & 31.98 (20.37) & 85.41 (73.80)\\
    \textbf{Graph Kernel} & 1241.11 & 3632.30 & 5213.42 & 13456.02\\
    \textbf{Supervised Embedding} & 1.01 (0.52) & 189.40 (188.91) & 754.35(753.86) & 950.68 (950.19) \\
    \bottomrule
	\end{tabular}
	}
	}
    \begin{tablenotes}
        \footnotesize
        \item[*] Numbers in brackets (if any) represent inference time
      \end{tablenotes}
    \end{threeparttable}
	\label{portion_efficiency}
\end{table}

%% file: 02related_work.tex
\section{Related Work}\label{sec:related}
In this section, we review the related work in terms of (1) team recommendations, (2) graph neural networks, and (3) graph node clustering.
\subsection{Team Recommendations}
Team recommendation problems concern with collaborations of individuals to complete a given task under certain constraints (e.g., limited communication budget). These problems take not only skill matches into account \cite{skills}, but also the connectivity among team members \cite{Lappas2009,Yin2018}. The balance in between has also been taken into serious consideration\cite{Dorn2010}. Based on these factors, contributions of individuals to the entire team have been carefully investigated \cite{Liangyue2017}. \textit{Team replacement} resolves the problem of team member replacement by applying graph kernels \cite{Liangyue2015} or reinforcement learning in a dynamic scenario~\cite{zhou2019towards}, where both skill features and collaboration structures are taken into account. With the problem extended to multiple team members to be replaced(\textit{subteam replacement}), graph kernel remained as the primary metric \cite{Zhaoheng2021}. Furthermore, existing methods ~\cite{Liangyue2015, Zhaoheng2021} produce fixed-size results, which might be problematic in a real-world scenario when new teams of smaller sizes are expected due to budget constraints. Recently, Zhou et al.~\cite{zhou2018extra} propose to interpret the results from graph kernel-based team recommendation algorithms with influence functions~\cite{zhou2019admiring,koh2017understanding,zhou2021adversarial}.
\textsc{Genius} is thus developed to resolve the limitations.
\subsection{Graph Neural Networks}
Various graph neural networks (GNN)\cite{GAT,GCN2016,GraphSAGE,GPRGNN} have been developed and widely applied in numerous fields. Within these models, node representations are passed through layers in a consecutive fashion~\cite{MessagePassing}. Each kind of them has its strengths and performs well in different tasks such as including social network analysis~\cite{GCN2016,GraphSAGE}, graph anomaly detection~\cite{ding2021few,ma2021comprehensive}. Graph convolutional networks (GCNs)~\cite{GCN2016} originated from the graph spectral theory and generate promising results by capturing information from the entire graph to encode each node. Graph Attention Networks(GAT)\cite{GAT} is able to learn the relationships of each node with its neighbors and GraphSAGE\cite{GraphSAGE} successfully deals with unseen nodes. Generalized page rank GNN (GPRGNN) \cite{GPRGNN} and Jumping Knowledge Net (JKNet) \cite{JKNet} further solves the problem of over-smoothing (i.e., nodes tend to have similar embeddings when the number of layers increases).
Regardless of various GNN structures, existing methods involve labels in the training process ~\cite{Cora,GCN2016,semi-supervised1}. \textit{Subteam replacement} calls for a model that generates more comprehensive team- and member-level representations without being constrained by labels.

\subsection{Graph Node Clustering}
Graph node clustering plays an important role in many graph pooling and coarsening techniques\cite{DMoN2020,DiffPool2018,MinCut2019}. Untrainable node clustering models have limited capability to extract important node features \cite{dhillon2007weighted}. Most trainable models are trained with ground-truth class labels \cite{DiffPool2018,SAGPool2019}. 
For models trained in an unsupervised fashion, some\cite{MinCut2019} are shown to be unable to optimize their own objective function on some real datasets\cite{DMoN2020}; others \cite{DMoN2020} have their partition validated over original labeled classes (i.e., clusters are based on class labels).
Since there is a very limited number of class labels (around 10)~\cite{DBLP,pitfalls,Cora}, nodes are not well separated (i.e., there are numerous nodes within a cluster) and thus, within-cluster search does not bring significant improvements to efficiency. Since it is challenging and unrealistic to involve numerous node labels, we adopt a novel unsupervised technique to train \textsc{Genius} and generate clusters with a small number of nodes.

%% file: 06conclusion.tex
\section{Conclusion}
In this paper, we investigate the challenging problem of \emph{subteam replacement}. To tackle this problem, we propose a novel framework ~\genius, which incorporates a GNN-backboned \emph{team network encoder} that learns essential team-level representations with \emph{positive team contrasting}, a training technique developed for team-related problems. To further improve the efficacy of ~\genius, members with similar skills and collaboration structures are clustered together, by which a \emph{subteam recommender} performs a within-cluster search that prunes unqualified candidates. Through extensive experiments, we demonstrate the superiority of \genius\ over existing methods. Our model can be further extended to multiple problems, including subgraph matching in large team social network datasets, subgraph similarities, etc.

%% file: 07ack.tex

%% file: 000main.bbl
\begin{thebibliography}{10}
\providecommand{\url}[1]{#1}
\csname url@samestyle\endcsname
\providecommand{\newblock}{\relax}
\providecommand{\bibinfo}[2]{#2}
\providecommand{\BIBentrySTDinterwordspacing}{\spaceskip=0pt\relax}
\providecommand{\BIBentryALTinterwordstretchfactor}{4}
\providecommand{\BIBentryALTinterwordspacing}{\spaceskip=\fontdimen2\font plus
\BIBentryALTinterwordstretchfactor\fontdimen3\font minus
  \fontdimen4\font\relax}
\providecommand{\BIBforeignlanguage}[2]{{%
\expandafter\ifx\csname l@#1\endcsname\relax
\typeout{** WARNING: IEEEtranS.bst: No hyphenation pattern has been}%
\typeout{** loaded for the language `#1'. Using the pattern for}%
\typeout{** the default language instead.}%
\else
\language=\csname l@#1\endcsname
\fi
#2}}
\providecommand{\BIBdecl}{\relax}
\BIBdecl

\bibitem{GED}
Z.~Abu-Aisheh, R.~Raveaux, J.-y. Ramel, and P.~Martineau, ``{An Exact Graph
  Edit Distance Algorithm for Solving Pattern Recognition Problems},'' in
  \emph{{4th International Conference on Pattern Recognition Applications and
  Methods 2015}}, 2015.

\bibitem{MinCut2019}
F.~M. Bianchi, D.~Grattarola, and C.~Alippi, ``Spectral clustering with graph
  neural networks for graph pooling,'' 2019.

\bibitem{Cora}
A.~Bojchevski and S.~Günnemann, ``Deep gaussian embedding of graphs:
  Unsupervised inductive learning via ranking,'' 2017.

\bibitem{randomwalk_nips}
K.~Borgwardt, N.~Schraudolph, and S.~Vishwanathan, ``Fast computation of graph
  kernels,'' in \emph{Advances in Neural Information Processing Systems}.\hskip
  1em plus 0.5em minus 0.4em\relax MIT Press, 2006.

\bibitem{ShortestPath}
\BIBentryALTinterwordspacing
K.~M. Borgwardt and H.-P. Kriegel, ``Shortest-path kernels on graphs,'' in
  \emph{Proceedings of the Fifth IEEE International Conference on Data Mining},
  ser. ICDM '05.\hskip 1em plus 0.5em minus 0.4em\relax USA: IEEE Computer
  Society, 2005, p. 74–81. [Online]. Available:
  \url{https://doi.org/10.1109/ICDM.2005.132}
\BIBentrySTDinterwordspacing

\bibitem{gcn_smoothing}
\BIBentryALTinterwordspacing
D.~Chen, Y.~Lin, W.~Li, P.~Li, J.~Zhou, and X.~Sun, ``Measuring and relieving
  the over-smoothing problem for graph neural networks from the topological
  view,'' 2019. [Online]. Available: \url{https://arxiv.org/abs/1909.03211}
\BIBentrySTDinterwordspacing

\bibitem{simclr}
\BIBentryALTinterwordspacing
T.~Chen, S.~Kornblith, M.~Norouzi, and G.~Hinton, ``A simple framework for
  contrastive learning of visual representations,'' 2020. [Online]. Available:
  \url{https://arxiv.org/abs/2002.05709}
\BIBentrySTDinterwordspacing

\bibitem{GPRGNN}
E.~Chien, J.~Peng, P.~Li, and O.~Milenkovic, ``Adaptive universal generalized
  pagerank graph neural network,'' 2020.

\bibitem{dhillon2007weighted}
I.~S. Dhillon, Y.~Guan, and B.~Kulis, ``Weighted graph cuts without
  eigenvectors a multilevel approach,'' \emph{IEEE transactions on pattern
  analysis and machine intelligence}, vol.~29, no.~11, pp. 1944--1957, 2007.

\bibitem{ding2021few}
K.~Ding, Q.~Zhou, H.~Tong, and H.~Liu, ``Few-shot network anomaly detection via
  cross-network meta-learning,'' in \emph{Proceedings of the Web Conference
  2021}, 2021, pp. 2448--2456.

\bibitem{Dorn2010}
C.~Dorn and S.~Dustdar, ``Composing near-optimal expert teams: A trade-off
  between skills and connectivity,'' 12 2010, pp. 472--489.

\bibitem{DBLP}
X.~Fu, J.~Zhang, Z.~Meng, and I.~King, ``{MAGNN}: Metapath aggregated graph
  neural network for heterogeneous graph embedding,'' in \emph{Proceedings of
  The Web Conference 2020}.\hskip 1em plus 0.5em minus 0.4em\relax {ACM}, apr
  2020.

\bibitem{MessagePassing}
J.~Gilmer, S.~S. Schoenholz, P.~F. Riley, O.~Vinyals, and G.~E. Dahl, ``Neural
  message passing for quantum chemistry,'' 2017.

\bibitem{networkx}
A.~A. Hagberg, D.~A. Schult, and P.~J. Swart, ``Exploring network structure,
  dynamics, and function using networkx,'' in \emph{Proceedings of the 7th
  Python in Science Conference}, Pasadena, CA USA, 2008, pp. 11 -- 15.

\bibitem{GraphSAGE}
W.~L. Hamilton, R.~Ying, and J.~Leskovec, ``Inductive representation learning
  on large graphs,'' 2017.

\bibitem{graphkit-learn}
\BIBentryALTinterwordspacing
L.~Jia, B.~Gaüzère, and P.~Honeine, ``graphkit-learn: A python library for
  graph kernels based on linear patterns,'' \emph{Pattern Recognition Letters},
  2021. [Online]. Available:
  \url{http://www.sciencedirect.com/science/article/pii/S0167865521000131}
\BIBentrySTDinterwordspacing

\bibitem{randomwalk_tong}
U.~Kang, H.~Tong, and J.~Sun, ``Fast random walk graph kernel,'' in
  \emph{Proceedings of the 12th SIAM International Conference on Data Mining,
  SDM 2012}, ser. Proceedings of the 12th SIAM International Conference on Data
  Mining, SDM 2012.\hskip 1em plus 0.5em minus 0.4em\relax Society for
  Industrial and Applied Mathematics Publications, 2012, pp. 828--838.

\bibitem{Marginalized}
H.~Kashima, K.~Tsuda, and A.~Inokuchi, ``Marginalized kernels between labeled
  graphs,'' vol.~1, 01 2003, pp. 321--328.

\bibitem{GCN2016}
T.~N. Kipf and M.~Welling, ``Semi-supervised classification with graph
  convolutional networks,'' 2016.

\bibitem{koh2017understanding}
P.~W. Koh and P.~Liang, ``Understanding black-box predictions via influence
  functions,'' in \emph{International conference on machine learning}.\hskip
  1em plus 0.5em minus 0.4em\relax PMLR, 2017, pp. 1885--1894.

\bibitem{Lappas2009}
T.~Lappas, K.~Liu, and E.~Terzi, ``Finding a team of experts in social
  networks,'' in \emph{Proceedings of the 15th ACM SIGKDD International
  Conference on Knowledge Discovery and Data Mining}, ser. KDD '09.\hskip 1em
  plus 0.5em minus 0.4em\relax New York, NY, USA: Association for Computing
  Machinery, 2009, p. 467–476.

\bibitem{SAGPool2019}
J.~Lee, I.~Lee, and J.~Kang, ``Self-attention graph pooling,'' 04 2019.

\bibitem{Liangyue2015}
L.~Li, H.~Tong, N.~Cao, K.~Ehrlich, Y.-R. Lin, and N.~Buchler, ``Replacing the
  irreplaceable: Fast algorithms for team member recommendation,'' 05 2015, pp.
  636--646.

\bibitem{Liangyue2017}
L.~Li, H.~Tong, Y.~Wang, C.~Shi, N.~Cao, and N.~Buchler, ``Is the whole greater
  than the sum of its parts?'' in \emph{Proceedings of the 23rd ACM SIGKDD
  International Conference on Knowledge Discovery and Data Mining}, ser. KDD
  '17.\hskip 1em plus 0.5em minus 0.4em\relax New York, NY, USA: Association
  for Computing Machinery, 2017, p. 295–304.

\bibitem{Zhaoheng2021}
Z.~Li, X.~Pi, M.~Wu, and H.~Tong, ``Reform: Fast and adaptive solution for
  subteam replacement,'' 2021.

\bibitem{ma2021comprehensive}
X.~Ma, J.~Wu, S.~Xue, J.~Yang, C.~Zhou, Q.~Z. Sheng, H.~Xiong, and L.~Akoglu,
  ``A comprehensive survey on graph anomaly detection with deep learning,''
  \emph{IEEE Transactions on Knowledge and Data Engineering}, 2021.

\bibitem{MarginalizedExtension}
\BIBentryALTinterwordspacing
P.~Mah\'{e}, N.~Ueda, T.~Akutsu, J.-L. Perret, and J.-P. Vert, ``Extensions of
  marginalized graph kernels,'' in \emph{Proceedings of the Twenty-First
  International Conference on Machine Learning}, ser. ICML '04.\hskip 1em plus
  0.5em minus 0.4em\relax New York, NY, USA: Association for Computing
  Machinery, 2004, p.~70. [Online]. Available:
  \url{https://doi.org/10.1145/1015330.1015446}
\BIBentrySTDinterwordspacing

\bibitem{pitfalls}
\BIBentryALTinterwordspacing
O.~Shchur, M.~Mumme, A.~Bojchevski, and S.~Günnemann, ``Pitfalls of graph
  neural network evaluation,'' 2018. [Online]. Available:
  \url{https://arxiv.org/abs/1811.05868}
\BIBentrySTDinterwordspacing

\bibitem{Grakel}
G.~Siglidis, G.~Nikolentzos, S.~Limnios, C.~Giatsidis, K.~Skianis, and
  M.~Vazirgiannis, ``Grakel: A graph kernel library in python,'' \emph{Journal
  of Machine Learning Research}, vol.~21, no.~54, pp. 1--5, 2020.

\bibitem{DMoN2020}
A.~Tsitsulin, J.~Palowitch, B.~Perozzi, and E.~Müller, ``Graph clustering with
  graph neural networks,'' 2020.

\bibitem{GAT}
P.~Veli{\v{c}}kovi{\'c}, G.~Cucurull, A.~Casanova, A.~Romero, P.~Lio, and
  Y.~Bengio, ``Graph attention networks,'' in \emph{ICLR}, 2018.

\bibitem{vishwanathan2010graph}
S.~V.~N. Vishwanathan, N.~N. Schraudolph, R.~Kondor, and K.~M. Borgwardt,
  ``Graph kernels,'' \emph{Journal of Machine Learning Research}, vol.~11, pp.
  1201--1242, 2010.

\bibitem{JKNet}
K.~Xu, C.~Li, Y.~Tian, T.~Sonobe, K.-i. Kawarabayashi, and S.~Jegelka,
  ``Representation learning on graphs with jumping knowledge networks,'' 2018.

\bibitem{semi-supervised1}
\BIBentryALTinterwordspacing
Z.~Yang, W.~W. Cohen, and R.~Salakhutdinov, ``Revisiting semi-supervised
  learning with graph embeddings,'' 2016. [Online]. Available:
  \url{https://arxiv.org/abs/1603.08861}
\BIBentrySTDinterwordspacing

\bibitem{Yin2018}
X.~Yin, C.~Qu, Q.~Wang, F.~Wu, B.~Liu, F.~Chen, X.~Chen, and D.~Fang, ``Social
  connection aware team formation for participatory tasks,'' \emph{IEEE
  Access}, vol.~6, pp. 20\,309--20\,319, 2018.

\bibitem{DiffPool2018}
R.~Ying, J.~You, C.~Morris, X.~Ren, W.~L. Hamilton, and J.~Leskovec,
  ``Hierarchical graph representation learning with differentiable pooling,''
  2018.

\bibitem{skills}
\BIBentryALTinterwordspacing
A.~ZAKARIAN and A.~KUSIAK, ``Forming teams: an analytical approach,'' \emph{IIE
  Transactions}, vol.~31, no.~1, pp. 85--97, 1999. [Online]. Available:
  \url{https://doi.org/10.1080/07408179908969808}
\BIBentrySTDinterwordspacing

\bibitem{zhou2018extra}
Q.~Zhou, L.~Li, N.~Cao, N.~Buchler, and H.~Tong, ``Extra: Explaining team
  recommendation in networks,'' in \emph{Proceedings of the 12th ACM Conference
  on Recommender Systems}, 2018, pp. 492--493.

\bibitem{zhou2019admiring}
Q.~Zhou, L.~Li, N.~Cao, L.~Ying, and H.~Tong, ``Admiring: Adversarial
  multi-network mining,'' in \emph{2019 IEEE International Conference on Data
  Mining (ICDM)}.\hskip 1em plus 0.5em minus 0.4em\relax IEEE, 2019, pp.
  1522--1527.

\bibitem{zhou2021adversarial}
------, ``Adversarial attacks on multi-network mining: Problem definition and
  fast solutions,'' \emph{IEEE Transactions on Knowledge and Data Engineering},
  2021.

\bibitem{zhou2019towards}
Q.~Zhou, L.~Li, and H.~Tong, ``Towards real time team optimization,'' in
  \emph{2019 IEEE International Conference on Big Data (Big Data)}.\hskip 1em
  plus 0.5em minus 0.4em\relax IEEE, 2019, pp. 1008--1017.

\end{thebibliography}
